\documentclass[11pt,letterpaper]{article}

\usepackage[utf8]{inputenc}
\usepackage[english]{babel}

\usepackage[dvipsnames,usenames]{xcolor}
\usepackage[colorlinks=true,pdfpagemode=UseNone,urlcolor=RoyalBlue,linkcolor=RoyalBlue,citecolor=OliveGreen,pdfstartview=FitH]{hyperref}

\usepackage{subcaption}

\usepackage[margin=1in]{geometry}
\usepackage{amssymb}
\usepackage{amsfonts,amsmath,amsthm,thmtools}
\usepackage{bm}
\usepackage{nicefrac}
\usepackage[ruled, vlined, linesnumbered]{algorithm2e}

\declaretheorem{theorem}
\declaretheorem[sibling=theorem]{lemma}
\declaretheorem[sibling=theorem]{corollary}
\declaretheorem[style=definition]{definition}
\declaretheorem[style=definition]{example}
\declaretheorem[sibling=theorem]{proposition}

\newcommand{\MMS}{\mathsf{MMS}}

\usepackage{color}
\usepackage{soul}

\usepackage{booktabs}
\usepackage{caption}

\usepackage[numbers, sort]{natbib}

\usepackage{tcolorbox}

\usepackage{enumitem}

\title{
   Maximin-Aware Allocations of Indivisible Chores  with Symmetric and Asymmetric Agents
}

\author{
    Tianze Wei
    \thanks{City University of Hong Kong. Email: t.z.wei-8@my.cityu.edu.hk}
    \and
    Bo Li
    \thanks{The Hong Kong Polytechnic University. Email: comp-bo.li@polyu.edu.hk}
    \and
    Minming Li
    \thanks{City University of Hong Kong. Email: minming.li@cityu.edu.hk.}   
}
\date{August 2023}

\begin{document}

\begin{titlepage}
    \thispagestyle{empty}
    \maketitle
    \begin{abstract}
        \thispagestyle{empty}
        The real-world deployment of fair allocation algorithms usually involves a heterogeneous population of users, which makes it challenging for the users to get complete knowledge of the allocation except for their own bundles. Chan et al. [IJCAI 2019] proposed a new fairness notion, maximin-awareness (MMA), which guarantees that every agent is not the worst-off one, no matter how the items that are not allocated to her are distributed. We adapt and generalize this notion to the case of indivisible chores and when the agents may have arbitrary weights. Due to the inherent difficulty of MMA, we also consider its up to one and up to any relaxations. A string of results on the existence and computation of MMA related fair allocations, and their connections to existing fairness concepts is given.
    \end{abstract}
\end{titlepage}

\section{Introduction}


Fairness is an important issue in many multi-agent systems, where all participants should be treated equally \cite{steinhaus1949division}. 
For example, network technology companies like Amazon, Google and Meta use schedulers in the cloud 
to allocate resources (e.g., servers, memory, etc.) or tasks (e.g., development, maintenance,  etc.) among a number of self-interested agents who want to maximize
the utility of their own allocations.
To ensure the sustainability of the Internet economy, the schedulers want the allocations to be fair \cite{DBLP:books/daglib/0017734,DBLP:conf/eurosys/VermaPKOTW15,DBLP:conf/sigcomm/GrandlAKRA14}.
\citet{books/daglib/0017730} presented the real-life applications of fair allocation problems, and a recent survey by \citet{amanatidis2023fair} reviewed the progress from the perspective of computer science and economics.

Two of the most well-established fairness criteria are \textit{envy-freeness} (EF) \cite{varian1974equity} and \textit{proportionality} (PROP) \cite{steinhaus1949division}. 
EF is an envy-based notion where every agent compares her own bundle with every other agent's and wants to get the best bundle. 
PROP is a share-based notion where every agent wants to ensure that the value of her bundle is no worse than $\frac{1}{n}$ fraction of her value for all the items where $n$ is the number of agents. 
When the items are indivisible, EF and PROP are hard to satisfy, and many relaxations have been studied instead in the literature. 
Among these relaxations, the ``up to one'' relaxation is one of the most popular ways, which requires the fairness notions to be satisfied after the removal of some item. The resulting notions are called EF or PROP up to one item, abbreviated as EF1 \cite{lipton2004approximately} and PROP1 \cite{conitzer2017fair}.
A stronger relaxation is ``up to any'' which strengthens the qualifier of the removed item to be arbitrary. Thus we get EF or PROP up to any item, abbreviated as EFX \cite{caragiannis2019unreasonable} and PROPX \cite{moulin2019fair}.
Besides these additive relaxations, maximin share fairness (MMS) \cite{budish2011combinatorial} is another popular notion, which requires every agent's value to be no worse than her worst share in an optimal $n$-partition of all items.
Formal definitions of these fairness notions are deferred to Section \ref{sec:prelim}.

\subsection{Epistemic Fair Allocation}

Motivated by the real-world applications where the agents do not have complete knowledge of the allocation due to privacy concerns or a huge number of agents involved in the system, an emerging line of research in fair resource allocation studies the epistemic fairness leveraging the information the agents have \cite{chen2017ignorance,aziz2018knowledge,DBLP:conf/aaai/HosseiniSVWX20,garg2022existence}. 
For example, epistemic EF was introduced by \citet{aziz2018knowledge} for the setting when the agents are connected via a social network and only know the bundles allocated to their neighbors. 
Informally, an allocation is epistemic EF if there exists a distribution of the items allocated to her non-neighboring agents such that the  allocation is EF to her.
Similar ideas have been extended to EFX by \citet{garg2022existence}.

What all the aforementioned works have in common is that they consider the epistemic variants of envy-based fairness notions. 
This is partly due to the fact that most share-based notions themselves do not require the knowledge of the allocations to the other agents. 
In comparison, \citet{chan2019maximin} proposed another epistemic fairness notion named \textit{maximin-aware} (MMA), which combines envy- and share-based requirements. 
Informally, the intuition of MMA is to ensure every agent does not obtain the worst bundle without knowing how the items that are not allocated to her are allocated among the other agents.
Since MMA is hard to satisfy, they also proposed the up to one and up to any relaxations. 

So far, most of the epistemic study of fair resource allocation has focused on the case of goods, when the agents want to obtain items with high value. 
Following \citet{chan2019maximin}, we study MMA allocations and make the following extensions. 
While \citet{chan2019maximin} only considered allocating goods among symmetric agents,
we study the allocation of indivisible chores (i.e., agents want to obtain items with low value) and the general situation where the agents may have different weights. 
We summarize our problem and the results in the following section.

\subsection{Our Problem and Main Contribution}

We study the maximin-aware (MMA) allocation of indivisible chores among $n$ agents whose cost functions are additive.
We use \textit{weights} to represent the asymmetry of agents in the
system when the agents may have different \textit{obligations} or \textit{responsibilities}
in a system. For example, a person in a leadership position is naturally expected to undertake higher collaboration responsibilities.
Weighted fairness has been justified
since the very early study of fair division in the context of
cake cutting problem \cite{robertson1998cake}.
In our problem, the agents know the number of agents, the weights of agents and the set of items. Meanwhile, they are only aware of their own bundles but do not know how the items that are not allocated to them are allocated among the other agents.
An MMA allocation guarantees that for each agent, there must exist some other agent whose bundle is no better than hers.
A bit more formally, an allocation is MMA if the cost of every agent's bundle is no greater than her $n-1$ (weighted) maximin share of the items that are not allocated to her, where 
weighted maximin share is defined in \citet{aziz2019weighted}.

\begin{figure}[t]
    \centering \includegraphics[width=0.5 \textwidth]{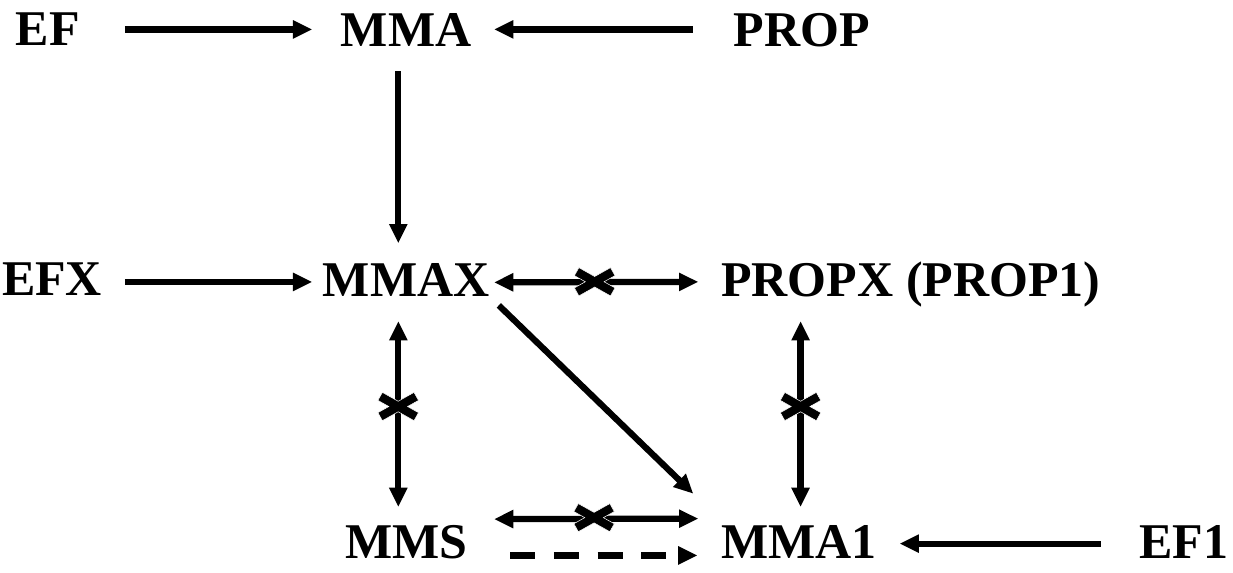}
    \caption{Relationship between MMA and other notions. The dotted arrow means the implication holds when the weights are the same. Here we do not show the known relationship such as EF implies EFX.}  \label{fig: MMA_Realtionship}
\end{figure}

Since MMA is hard to satisfy, we also consider its ``up to one'' and ``up to any'' relaxations, resulting in the notions of MMA1 and MMAX.
The relationships between MMA1/MMAX and existing notions are shown in Figure \ref{fig: MMA_Realtionship}.
We can see that general EF related notions are stronger than MMA related ones, but PROP related notions and MMA related ones are not comparable. 
An exception is that MMS implies MMA1 when the agents have the same weight, but fails to ensure any approximation of MMA1 if the agents' weights are arbitrary. 
Another one is that a PROP allocation is also MMA;
however, a PROP1 or PROPX allocation does not have any guaranteed approximation of MMA1 or MMAX.

The advantage of MMA1 is its guaranteed existence for agents with arbitrary weights.
Regarding MMAX, we show its existence when the agents have the same weight. When the agents have different weights, MMAX allocations may not exist for two agents, implied by \cite{hajiaghayi2023almost}.
When the agents have arbitrary weights, we design  a 1.91-approximation algorithm if $n=2$ and a $(1+\lambda)$-approximation algorithm, where $\lambda = \frac{3-n+\sqrt{n^2+10n-7}}{4n-4}$, if $n\geq 3$. 
We plot the approximation ratios in Figure \ref{fig: Approximation_ratio}. As we can see, since $\frac{1}{n-1} < \lambda < \frac{2}{n-1}$, the approximation ratio becomes close to 1 when $n$ becomes large. 
Designing algorithms for envy-based notions (such as EFX) is arguably harder than that for share-based ones (such as PROPX) since it requires comparisons between every agent with every other agent. 
For example, a PROPX allocation always exists when the agents have arbitrary weights \cite{li2022almost}; however, the best known approximation of EFX is $O(n^2)$ for agents with the same weight \cite{DBLP:conf/ijcai/0002022}, and nothing is known beyond the same weight. 
\begin{figure}[t]
    \centering \includegraphics[width=0.5 \textwidth]{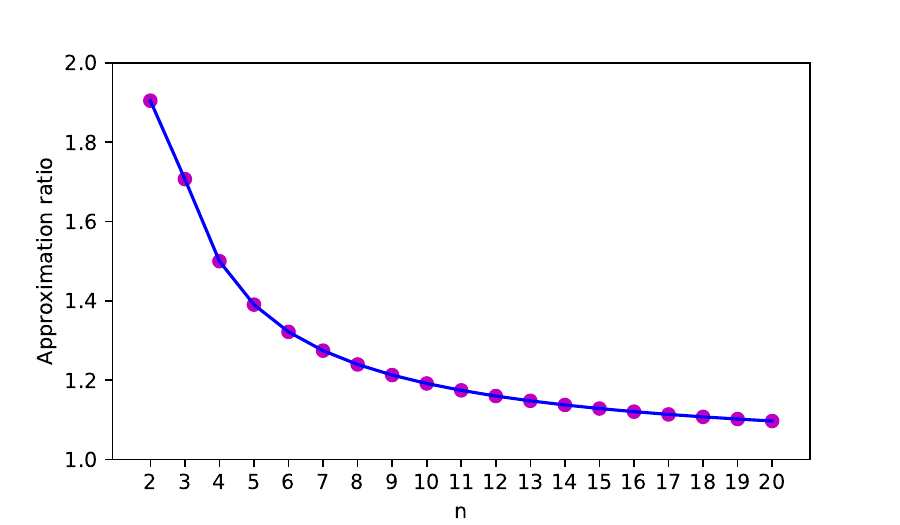}
    \caption{The illustration of the approximation ratios for MMAX when agents have arbitrary weights.}  \label{fig: Approximation_ratio}
\end{figure}

\subsection{Other Related Work}
\label{sec: related work}
Fair allocation of indivisible items for the case of goods has been widely studied; see the survey of \citet{amanatidis2023fair} for an overview. There is also a recent line of work regarding the fair allocation of indivisible chores. \citet{aziz2017algorithms} initiated the fair allocation of chores with MMS notions and showed that there exists an instance where MMS allocations do not exist. But approximate MMS allocations can be computed efficiently \citet{aziz2017algorithms,barman2020approximation,huang2021algorithmic,huang2023reduction}.
When the agents are asymmetric, \citet{aziz2019weighted} and \citet{DBLP:journals/corr/abs-2211-13951} explored the weighted MMS fairness and anyprice share (APS) fairness.
\citet{aziz2020polynomial} proposed a polynomial-time algorithm to compute an allocation that satisfies Pareto optimality and PROP1 simultaneously for asymmetric agents when the set of items includes goods and chores. \citet{sun2021connections} studied the connections among several fairness notions like EFX, MMS, etc. in allocating chores. 
\citet{li2022almost} and \citet{wu2023weighted}, respectively, showed that when agents are symmetric, PROPX and EF1 allocations exist and can be computed in polynomial time.
\citet{garg2022fair} presented a polynomial-time algorithm to compute an allocation that satisfies Pareto optimality and EF1 simultaneously when agents have at most two values for chores, and \citet{wu2023weighted} extended this result to asymmetric agents.
\citet{hajiaghayi2023almost} showed that EFX allocations do not always exist when there are two or three agents with different weights.

\section{Preliminaries}
\label{sec:prelim}
We introduce our model and solution concepts in this section. 
Let $N=\{1,\ldots,n\}$ be a set of $n$ agents, and $M = \{f_1,\ldots,f_m\}$ be a set of $m$ indivisible chores. Each agent $i \in N$ has a cost function $c_{i}: 2^{N} \rightarrow \mathbb{R}_{\geq 0}$. The cost functions are assumed to be additive, i.e., for any $S\subseteq M$, $c_i(S) = \sum_{f\in S} c_i(\{f\})$. 
For simplicity, we use $c_i(f)$ instead of $c_i(\{f\})$ for $f \in M$. 
We study the weighted setting, where each agent has a weight $w_i > 0$, and the weights add up to one, i.e., $\sum_{i \in N}w_i =1$. 
Without loss of generality, we assume that $c_i(M) = 1$ for all agents $i \in N$.
A chores allocation instance is denoted as $\mathcal{I} = \langle M, N, \boldsymbol{c}, \boldsymbol{w}\rangle$, where $\boldsymbol{c} = (c_1,\ldots,c_n)$ and $\boldsymbol{w}= (w_1,\ldots,w_n)$. 
An allocation $\mathcal{X} = (X_1, \ldots, X_n)$, where $X_i$ is the bundle allocated to agent $i$, is an $n$-partition of $M$ among $n$ agents, i.e., $\bigcup_{i \in N}X_{i} = M$ and $X_{i} \cap X_{j} = \emptyset$ for any two agents $i \neq j$. 
Let $\Pi_{n}(M)$ denote the set of all $n$-partitions of $M$.
Particularly, we denote the set of items that are not allocated to agent $i$ as $X_{-i} = \bigcup_{j \in N \setminus \{i\}}X_j$.


\subsection{Fairness Concepts}

We next introduce the most classic fairness notions, including \textit{envy-freeness},  \textit{proportionality} and their relaxations.

\begin{definition}[\textbf{EF}]
For any $\alpha \in [1,+\infty)$, an allocation $\mathcal{X} = (X_1,\ldots,X_n)$ is $\alpha$-approximate envy-free ($\alpha$-EF), if for any two agents $i, j \in N$, 
$\frac{c_i(X_i)}{w_i}\leq \alpha \cdot \frac{c_i(X_j)}{w_j}.$
The allocation is EF if $\alpha=1$.
\end{definition}

\begin{definition}[\textbf{EF1 and EFX}]
For any $\alpha \in [1,+\infty)$, an allocation $\mathcal{X} = (X_1,\ldots,X_n)$ is 
$\alpha$-approximate envy-free up to one item ($\alpha$-EF1), if for any two agents $i, j \in N$ with $X_i \neq \emptyset$, 
\begin{equation}
\label{eq:ef1}
    \frac{c_i(X_i \setminus \{f\})}{w_i}\leq \alpha \cdot \frac{c_i(X_j)}{w_j} \text{ for some item $f \in X_i$}.
\end{equation}
If the quantifier ``some'' in Inequality (\ref{eq:ef1}) is changed to ``any'', the allocation is $\alpha$-approximate envy-free up to any item ($\alpha$-EFX).
The allocation is EF1 or EFX if $\alpha=1$.

\end{definition}


\begin{definition}[\textbf{PROP}]
For any $\alpha \in [1,+\infty)$, an allocation $\mathcal{X}=(X_1,\ldots,X_n)$ is $\alpha$-approximate proportional ($\alpha$-PROP), if for any agent $i \in N$, 
$c_i(X_i) \leq \alpha \cdot w_i \cdot c_i(M).$
The allocation is PROP if $\alpha=1$.
\end{definition}

\begin{definition}[\textbf{PROP1 and PROPX}]
\label{def: prop1_propx_chores}
For any $\alpha \in [1,+\infty)$, an allocation $\mathcal{X}=(X_1,\ldots,X_n)$ is $\alpha$-approximate proportional up to one item ($\alpha$-PROP1), if for any agent $i \in N$ with $X_i \neq \emptyset$, 
\begin{equation}
\label{eq:prop1}
    c_i(X_i \setminus \{f\}) \leq \alpha \cdot w_i \cdot c_i(M) \text{ for some item $f \in X_i$.}
\end{equation}
If the quantifier ``some'' in Inequality (\ref{eq:prop1}) is changed to ``any'', the allocation is $\alpha$-approximate proportional up to any item ($\alpha$-PROPX).
The allocation is PROP1 or PROPX if $\alpha=1$.


\end{definition}

Besides PROP1 and PROPX, \textit{maximin share fairness} is another popular relaxation of \textit{proportionality}.

Given an instance $\mathcal{I} = \langle M, N, \boldsymbol{c}, \boldsymbol{w} \rangle$, the \textit{maximin share} of agent $i$ on $M$ among $n$ agents is defined as :
$$ \MMS_{i}(M, n) =  w_i \cdot \min_{ \mathcal{Y} \in \Pi_{n}(M)} \max_{j \in N} \frac{ c_{i}(Y_j)}{w_j}.$$

Different from the unweighted case, we can see that even if two agents have the same cost function, they may still have different MMS values.

\begin{definition}[\textbf{MMS}]
For any $\alpha \in [1,+\infty)$, an allocation $\mathcal{X}=(X_1,\ldots,X_n)$ is $\alpha$-approximate maximin share fair ($\alpha$-MMS), if for any agent $i \in N$, $c_i(X_i) \leq \alpha \cdot \MMS_{i}(M, n).$ The allocation is MMS if $\alpha=1$.
\end{definition}

\subsection{Maximin-Aware and Its Relaxations}

Next, we introduce our main solution concept, maximin-aware fairness, which is a hybrid notion of EF and MMS. 
Intuitively, 
a bundle $X_i$ is maximin-aware (MMA) fair to an agent $i \in N$ if no matter how the items not allocated to her are distributed among the other agents, there always exists one agent whose  bundle is no better than hers. 
In other words, for an arbitrary $(n-1)$-allocation $(Y_j)_{j\neq i}$ of $X_{-i}$, there exists $j^* \in N \setminus \{i\}$ such that
$
\frac{c_i(X_i)}{w_i} \le \frac{c_i(Y_{j^*})}{w_{j^*}}.
$
Equivalently, a bundle $X_i$ is MMA fair to an agent $i \in N$ if $c_i(X_i) \le \MMS_{i}(X_{-i}, n-1)$, where 
\[
\MMS_{i}(X_{-i}, n-1) = w_i \cdot \min_{\mathcal{Z} \in  \Pi_{n-1}(X_{-i})} \max_{j \in N \setminus \{i\}} \frac{c_{i}(Z_j)}{w_j}.
\]


\begin{definition}[\textbf{MMA}]

For any $\alpha \in [1,+\infty)$, an allocation $\mathcal{X}=(X_1,\ldots,X_n)$ is $\alpha$-approximate maximin-aware fair ($\alpha$-MMA), if for any agent $i \in N$, 
$
c_i(X_i) \leq \alpha \cdot \MMS_{i}(X_{-i}, n-1).
$
The allocation is MMA if $\alpha=1$.
\end{definition}


Similar to EF and MMS, MMA allocations may not exist.
Suppose that there are two agents and one item with positive costs for both agents. In any allocation, there always exists one agent who cannot satisfy the condition of MMA. Thus, we consider the relaxations of MMA.

\begin{definition}[\textbf{MMA1 and MMAX}]

For any $\alpha \in [1,+\infty)$, an allocation $\mathcal{X}=(X_1,\ldots,X_n)$ is $\alpha$-approximate maximin-aware up to one item ($\alpha$-MMA1), if for any agent $i \in N$ with $X_i \neq \emptyset$, there exists one item $f \in X_{i}$ such that 
\begin{equation}
\label{eq:mma1}
    c_i(X_i \setminus \{f\}) \leq \alpha \cdot \MMS_{i}(X_{-i}, n-1).
\end{equation}
Similarly, the allocation is $\alpha$-approximate maximin-aware up to any item ($\alpha$-MMAX) if Inequality (\ref{eq:mma1}) holds for any item $f \in X_{i}$. The allocation is MMA1 or MMAX if $\alpha=1$.
\end{definition}

By the definitions, it is straightforward that any $\alpha$-MMA allocation is $\alpha$-MMAX, and any $\alpha$-MMAX allocation is $\alpha$-MMA1.
Next, we illustrate these definitions via an example. 



\begin{example}
Consider an instance with three agents and five items. Assume that their weights are $w_1 = \frac{1}{2}$, $w_2 = \frac{1}{3}$ and $w_3 = \frac{1}{6}$ respectively. For simplicity, assume that they have the same cost function, as shown in Table \ref{tab:example_mma}. 

 \begin{table}[h]
     \centering
     \begin{tabular}{c|ccccc}
     \toprule
    $f_j$ & $f_1$ & $f_2$ & $f_3$ & $f_4$ & $f_5$ \\
    \midrule
     $c(f_j)$ & $\frac{19}{72}$ & $\frac{17}{72}$ & $\frac{2}{9}$ & $\frac{11}{72}$ & $\frac{1}{8}$ \\
     \bottomrule
     \end{tabular}
     \caption{An example of MMA related notions.}
     \label{tab:example_mma}
 \end{table}
 
Let us examine an allocation $\mathcal{X} = (X_1,X_2,X_3)$ with $X_1 = \{f_3\}$, $X_2 = \{f_1,f_2\}$ and $X_3 = \{f_4,f_5\}$. It is not hard to see that for agent 1, $X_1$ is MMA to her since $c(X_1) = \frac{2}{9} < \MMS_1(X_{-1},2) = \max\{(\frac{17}{72}+\frac{11}{72}+\frac{1}{8}) \times \frac{3}{2}, \ \frac{19}{72} \times 3 \} = \frac{19}{24}$; for agent 2, $X_2$ is MMA1 but not MMAX to her since $c(X_2 \setminus \{f_1\}) = \frac{17}{72} < \MMS_2(X_{-2},2) = \max \{ (\frac{2}{9}+ \frac{11}{72}) \times \frac{2}{3}, \ \frac{1}{8} \times 2 \} = \frac{1}{4}$ and $c(X_2 \setminus \{f_2\}) = \frac{19}{72} > \MMS_2(X_{-2},2) = \frac{1}{4}$; for agent 3, $X_3$ is MMAX to her since $c(X_3 \setminus \{f_5\}) = \frac{11}{72} = \MMS_3(X_{-3},2) = \max \{(\frac{16}{72} + \frac{17}{72}) \times \frac{1}{3}, \ \frac{19}{72} \times \frac{1}{2}\}= \frac{11}{72}$ and $c(X_3 \setminus \{f_5\}) > c(X_3 \setminus \{f_4\})$.

\end{example}


\paragraph{Ordered instance.} An instance $\mathcal{I} = \langle M, N, \boldsymbol{c}, \boldsymbol{w} \rangle$ is called {\em ordered} if all agents have the same ranking on all the items, i.e., for every agent $i\in N$, 
\[
c_i(f_1) \ge c_i(f_2) \ge \cdots \ge c_i(f_m).
\]
Note that in an ordered instance, the agents can still have different cardinal values for the items. 
Intuitively, there are more conflicts among agents in an ordered instance than in a general one since agents desire the same item. 
In fact, \cite{bouveret2016characterizing} and \cite{barman2020approximation} formalized this intuition and showed that any algorithm that ensures $\alpha$-MMS allocations for ordered instances can be converted to an algorithm for the general instances in polynomial time with the same guarantee. 
We show that the same result holds for MMA1 and MMAX allocations even when the agents have non-identical weights. 

\begin{definition}[\textbf{Ordered instance}]
\label{def: ordered_instance}

Given an instance $\mathcal{I} = \langle M, N, \mathbf{c}, \mathbf{w} \rangle$, the corresponding ordered instance of $\mathcal{I}$ is denoted as $\mathcal{I}' = \langle M, N, \mathbf{c'}, \mathbf{w} \rangle$, where for each agent $i \in N$ and item $f_{j} \in M$, we have $c_i^{'}(f_j) = c_i(f^{j})$ ($f^{j}$ is the $j^{th}$ highest cost item of $M$ for agent $i$).
\end{definition}

\begin{algorithm}
\caption{Reduction}
\label{alg: reduction_alg_chores}
\KwIn{An instance $\mathcal{I}  = \langle M, N, \mathbf{c}, \mathbf{w} \rangle $, its ordered instance $\mathcal{I}^{'}  = \langle M, N, \mathbf{c'}, \mathbf{w} \rangle$ and an  allocation $\mathcal{X}^{'} = (X_1^{'},\ldots, X_n^{'})$ of ordered instance $\mathcal{I^{'}}$}
\KwOut{An allocation $\mathcal{X}=(X_1,\ldots,X_n)$ of instance $\mathcal{I}$}

Let $\mathcal{X} = (\emptyset, \ldots, \emptyset)$ and $T = M$;

\For(\tcp*{From the smallest item to the largest one}){$j = m $ \KwTo $1$}{Let $i^{*}$ denote the agent whose bundle $X_i^{'}$ includes $f_j$;

Choose the item $f_{j}^{*}$ = $\arg \min_{f \in T} c_{i}(f)$; \label{lin: pick_target_chore}

$X_{i^{*}} = X_{i^{*}} \cup \{f_{j}^{*}\}$ and $T = T \ \backslash \ \{f_{j}^{*}\}$;
}

\Return $\mathcal{X}$
\end{algorithm}

\begin{lemma}
\label{lem: property_reduction_chores}
 Given an instance $\mathcal{I}$, the corresponding ordered instance $\mathcal{I}^{'}$ and an allocation $\mathcal{X}^{'}=(X_1^{'},\ldots,X_n^{'})$ of $\mathcal{I}^{'}$, the allocation $\mathcal{X}=(X_1,\ldots,X_n)$ returned by Algorithm \ref{alg: reduction_alg_chores} has the following properties:
\begin{enumerate}
    \item $\lvert X_i \rvert = \lvert X_i^{'} \rvert$.
    \item $c_i(X_i \backslash \{f\}) \leq c_i^{'}(X_i^{'} \backslash \{f^{'}\})$, where $f \in X_{i}$, $f^{'} \in X_{i}^{'}$ and $f$ is the corresponding chosen item in line \ref{lin: pick_target_chore} when $f^{'}$ is considered.
    \item $\MMS_i(X_{-i},n-1) \geq \MMS_i^{'}(X_{-i}^{'},n-1)$, where $\MMS_{i}^{'}(X_{-i}^{'}, n-1)$ is defined by the cost function $c_{i}^{'}$.
\end{enumerate}
\end{lemma}

\begin{proof}
For the first statement, consider the round of $j$ of the for-loop in Algorithm \ref{alg: reduction_alg_chores}; only agent $i$ who picks $f_j$ can choose what she likes most from the remaining available items $T$. Thus, it is not hard to see that $\lvert X_i \rvert = \lvert X_i^{'} \rvert$. 

To prove the second statement, it is sufficient to prove when $f^{*}_{j}$ is selected from $T$ in line \ref{lin: pick_target_chore} of Algorithm \ref{alg: reduction_alg_chores}, we have $c_i(f_{j}^{*}) \leq c_{i}^{'}(f_{j})$. Because if it is true, for any agent $i \in N$, there are $\lvert X_{i} \rvert $ inequalities that hold. After we add arbitrary $\lvert X_i \rvert -1$ inequalities up, we can get the desired inequality. Next, we consider the occasion when agent $i$ picks $f_j$ in the round of $j$. Notice that, before the round of $j$, exactly $m-j$ items are chosen, and there are $j$ items left in $T$. Therefore, the worst item that agent $i$ can obtain from $T$ is the one that has the same cost as $c_{i}^{'}(f_j)$, i.e., we have $c_{i}(f_{j}^{*}) \leq c_{i}^{'}(f_{j})$.

For the simplicity of the proof of the last statement, all items $\{f_{(k)}\}$ in $X_{-i}$ and $\{f_{(k)}^{'}\}$ in $X_{-i}^{'}$, where $k = 1, \ldots, \lvert X_{-i} \rvert$, are assumed to be sorted in the increasing order of cost respectively. By the above two statements, we have $\lvert X_i \rvert = \lvert X_i^{'} \rvert$, and $c_i(f) \leq c_i^{'}(f^{'})$, where $f \in X_i$, $f^{'} \in X_{i}^{'}$ and $f$ is the corresponding chosen item of $f^{'}$ in line \ref{lin: pick_target_chore}. Therefore, it is easy to see that $\lvert X_{-i} \rvert = \lvert X_{-i}^{'} \rvert$ and for each pair of items $f_{(k)}$ and $f_{(k)}^{'}$, where $k = 1, \ldots, \lvert X_{-i} \rvert$, from $X_{-i}$ and $X_{-i}^{'}$ respectively, we have $c_i(f_{(k)}) \geq c_{i}^{'}(f_{(k)}^{'})$.

 Suppose, towards a contradiction, that $\MMS_i(X_{-i},n-1) < \MMS_i^{'}(X_{-i}^{'},n-1)$. Let $\mathcal{Y} = (Y_j)_{j \neq i}$, where $w_i\cdot \max_{j \in N \backslash\{i\} }\frac{c_i(Y_j)}{w_j} = \MMS_i(X_{-i},n-1)$, be an $(n-1)$-allocation of $X_{-i}$. Next, if we replace each item in $\mathcal{Y}$ with the item with the same index from $X_{-i}^{'}$, we obtain an allocation $\mathcal{Y}^{'}$ that is an $(n-1)$-partition of $X_{-i}^{'}$. In this allocation $\mathcal{Y}^{'}$, we have $c_i(Y_j) \geq c_i^{'}(Y_{j}^{'})$ for any $j \in N \backslash \{i\}$. So we get $w_{i} \cdot \max_{j \in N \backslash\{i\}} \frac{c_i^{'}(Y_{j}^{'})}{w_{j}} \leq w_i \cdot \max_{j \in N \backslash\{i\}} \frac{c_{i}(Y_{j})}{w_{j}} = \MMS_{i}(X_{-i},n-1) < \MMS_i^{'}(X_{-i}^{'},n-1)$, which contradicts the definition of $\MMS_i^{'}(X_{-i}^{'},n-1)$.
\end{proof}

\begin{lemma}
\label{lem: wmma1_reduction_approximation_chores}
For any $\alpha \in [1,+\infty)$, suppose that there is a polynomial time algorithm that computes an $\alpha$-MMA1 allocation $\mathcal{X}^{'}=(X_1^{'}, \ldots, X_n^{'})$ of an ordered instance $\mathcal{I}^{'}$, then we have a polynomial time algorithm that computes an $\alpha$-MMA1 allocation $\mathcal{X} = (X_1, \ldots, X_n)$ of instance $\mathcal{I}$.
\end{lemma}
\begin{proof}
   For any agent $i \in N$, by Lemma \ref{lem: property_reduction_chores}, we have  $c_i(X_i \backslash \{f_{max}\}) \leq c_i^{'}(X_i^{'} \backslash \{f_{max}^{'}\})$, where $f_{max} =\arg \max_{f \in X_i}c_i(f)$ and $f_{max}^{'} =\arg\max_{f \in X_i^{'}}c_{i}^{'}(f)$, and $\MMS_i(X_{-i},n-1) \geq \MMS_i^{'}(X_{-i}^{'},n-1)$. Therefore, if we have $c_{i}^{'}(X_i^{'} \backslash \{f_{max}^{'}\}) \leq \alpha \cdot \MMS_i^{'}(X_{-i}^{'},n-1)$, $c_i(X_i \backslash \{f_{max}\}) \leq \alpha \cdot \MMS_i(X_{-i},n-1)$ also holds.
\end{proof}

The proof of the following lemma is identical to that of Lemma \ref{lem: wmma1_reduction_approximation_chores} except that we use
$f_{min}$ and $f_{min}^{'}$ to replace $f_{max}$ and $f_{max}^{'}$ respectively, where $f_{min} = \arg \min_{f \in X_i}c_i(f)$ and $f_{min}^{'} = \arg \min_{f \in X_{i}^{'}}c_i^{'}(f)$.

\begin{lemma}
\label{lem: wmmax_reduction_approximation_chores}
For any $\alpha \in [1,+\infty)$, suppose that there is a polynomial time algorithm that computes an $\alpha$-MMAX allocation $\mathcal{X}^{'}=(X_1^{'}, \ldots, X_n^{'})$ of an ordered instance $\mathcal{I}^{'}$, then we have a polynomial time algorithm that computes an $\alpha$-MMAX allocation $\mathcal{X} = (X_1, \ldots, X_n)$ of instance $\mathcal{I}$.
\end{lemma}

It is common that the reduction from an arbitrary instance to an ordered instance is used in MMS or PROPX allocations where these two fairness notions emphasize the comparison between the value of a bundle and a fixed threshold value; see \citep{barman2020approximation, garg2020improved,huang2021algorithmic,li2022almost}. Surprisingly, we show the reduction also works in MMA1 and MMAX allocations. What is the difference in terms of proof? In MMS or PROPX allocations, it is sufficient to show that after the execution of the reduction algorithm, for any agent $i \in N$, we have $c_i(X_i) \leq c_{i}^{'}(X_i^{'})$ or $c_i(X_i \backslash \{f_{min}\}) \leq c_{i}^{'}(X_{i}^{'} \backslash \{f_{min}^{'}\})$ where $f_{min} = \arg \min_{f \in X_i}c_i(f)$ and $f_{min}^{'} = \arg \min_{f \in X_{i}^{'}}c_{i}^{'}(f)$  since the benchmark values are the same in the original instance and its ordered instance, i.e., $\MMS_{i}(M,n) = \MMS_{i}^{'}(M,n)$ and $w_i \cdot c_i(M) = w_i \cdot c_{i}^{'}(M)$. However, it is insufficient in MMA1 and MMAX allocations, and it is necessary to prove the following two conditions. On the one hand, we need to show $c_i(X_i \backslash \{f\}) \leq c_i^{'}(X_i^{'} \backslash \{f^{'}\})$ where $f$ and $f^{'}$ are the items with the maximum cost or minimum cost in $X_i$ and $X_{i}^{'}$ respectively; On the other hand; we need to show $\MMS_{i}(X_{-i},n-1) \geq \MMS_{i}^{'}(X_{-i}^{'},n-1)$. Only when the above two conditions are satisfied simultaneously can the reduction work.

Therefore, in the following sections, we only focus on ordered instances.




\section{Connections between MMA and Other Fairness Notions}
In this section, we introduce the connections between MMA related notions and the ones related to EF, PROP and MMS.

\subsection{EF, EF1 and EFX}

We start with EF related notions. 
By the definitions, it is not hard to verify that EF implies MMA, and the implication also holds for up to one/any relaxations. 

\begin{lemma}
For any $\alpha \in [1,+\infty)$, any $\alpha$-EF allocation is also $\alpha$-MMA.    
\end{lemma}

\begin{lemma}
\label{lem:ef1_mma1_chores}
For any $\alpha \in [1,+\infty)$, any $\alpha$-EF1 allocation is also $\alpha$-MMA1. 
\end{lemma}

\begin{proof}
 Let $\mathcal{X} = (X_1,\ldots, X_n)$ be an $\alpha$-EF1 allocation. By the definition of EF1 allocations, for any agent $i \in N$, we have $w_{j} \cdot c_i(X_i \setminus \{f_{max}^{i}\}) \leq \alpha \cdot w_{i} \cdot c_i(X_j)$, where $f_{max}^{i} = \arg \max_{f \in X_i}c_i(f)$, for any agent $j \in N \setminus \{i\}$. Then, summing up respective inequalities for all $ j \in N \setminus \{i\}$, we get 
\begin{equation}
(1-w_i)\cdot c_i(X_i \setminus \{f_{max}^{i}\}) \leq \alpha \cdot w_i \cdot c_i(X_{-i}).
\label{eqn: ef1_mma1_chore_inequality}  
\end{equation}
Suppose, for the contradiction, that $\mathcal{X}$ is not an $\alpha$-MMA1 allocation, i.e., there exists some agent $i \in N $ such that $c_i(X_i \setminus \{f\}) > \alpha \cdot \MMS_{i}(X_{-i},n-1)$ holds for any item $f \in X_{i}$. Let $\mathcal{Y} = (Y_j)_{j \neq i}$, where $w_i\cdot \max_{j \in N \setminus \{i\}}\frac{c_i(Y_j)}{w_j} = \MMS_i(X_{-i},n-1)$, be an $(n-1)$-allocation of $X_{-i}$. So we have $w_j \cdot c_i(X_i \setminus \{f_{max}^{i}\}) > \alpha \cdot w_i \cdot c_i(Y_j)$ for any $j \in N \setminus \{i\}$. Similarly, summing up respective inequalities for all $j \in N \setminus \{i\}$, we have $(1-w_i) \cdot c_i(X_i \setminus \{f_{max}^{i}\}) > \alpha \cdot w_i \cdot c_i(X_{-i})$. which contradicts Inequality (\ref{eqn: ef1_mma1_chore_inequality}). 
\end{proof}

\citet{wu2023weighted} showed that a (weighted) EF1 allocation can be computed in polynomial time, and thus by Lemma \ref{lem:ef1_mma1_chores}, we directly have the following corollary. 
    
\begin{corollary}
\label{coro:MMA1}
    MMA1 allocations exist and can be computed in polynomial time.
\end{corollary}


The proof of the following lemma is identical to that of Lemma \ref{lem:ef1_mma1_chores}. except that we use
$f_{min}^{i}$ to replace $f_{max}^{i}$ where $f_{min}^{i} = \arg \min_{f \in X_i}c_i(f)$.

\begin{lemma}
\label{lem:efx_mmax_chores}
For any $\alpha \in [1,+\infty)$, any $\alpha$-EFX allocation is also $\alpha$-MMAX. 
\end{lemma}

When agents have different weights, EFX allocations may not exist for two or three agents \cite{hajiaghayi2023almost}. Especially, when there are two agents, EFX is equivalent to MMAX. Thus, MMAX allocations may not exist for two agents. However, when the agents have the same weight, the \textit{Top-trading Envy Cycle Elimination} algorithm \cite{li2022almost,DBLP:conf/approx/BhaskarSV21} can find EFX allocations for ordered instances in polynomial time. By Lemma \ref{lem:efx_mmax_chores}, we have the following corollary.

\begin{corollary}
    When all agents have the same weight, MMAX allocations exist and can be computed in polynomial time.
\end{corollary}
 

\subsection{PROP, PROP1 and PROPX}

Next, we discuss the connections between MMA and PROP related notions. 
As we will see, in sharp contrast to EF, although PROP still implies MMA, PROP1 and PROPX do not guarantee any bounded approximation for MMA1 and MMAX. 


\begin{lemma}
\label{lem: ef_prop_mma_chores}
Any PROP allocation is also MMA.
\end{lemma}

\begin{proof}
Let $\mathcal{X} = (X_1,\ldots, X_n)$ be a PROP allocation. By the definition of PROP allocations, for any agent $i \in N$, we have $c_i(X_i) \leq {w_i}$. Suppose, for the contradiction, that $\mathcal{X}$ is not an MMA allocation, i.e., there exists some agent $i \in N$ such that $c_i(X_i) > \MMS_{i}(X_{-i},n-1)$. Let $\mathcal{Y} = (Y_j)_{j \neq i}$, where $w_i\cdot \max_{j \in N \backslash \{i\}}\frac{c_i(Y_j)}{w_j} = \MMS_i(X_{-i},n-1)$, be an $(n-1)$-allocation of $X_{-i}$. So we have $w_{j} \cdot c_i(X_i) > w_i \cdot c_i(Y_j) $ for any $j \in N \backslash \{i\}$. Then, summing up respective inequalities for all $j \in N \backslash \{i\}$, we get $(1-w_i)\cdot c_i(X_i) > w_i \cdot c_i(X_{-i})$. At last, adding $w_i \cdot c_i(X_i)$ to both sides of the above inequality, we have $c_i(X_i) > w_i$, which contradicts that $\mathcal{X}$ is PROP.
\end{proof}




\begin{proposition}
\label{pro:propx_mma1_mmax_chores}
 A PROPX allocation is not necessarily $\alpha$-MMA1 or $\alpha$-MMAX for any $\alpha \in [1,+\infty)$, even when the agents have the same weight.   
\end{proposition}
\begin{proof}

 Consider an instance with two agents with identical weights and four items. We focus on agent $1$, and the cost of each item, according to agent 1, is shown in Table \ref{tab: propx_mma1_mmax_example}.  
 
 \begin{table}[h]
     \centering
     \begin{tabular}{c|cccc}
     \toprule
    $f_j$ & $f_1$ & $f_2$ & $f_3$ & $f_4$\\
    \midrule
     $c_1(f_j)$ & $\frac{1}{2} - \epsilon$ & $\frac{1}{2} - \epsilon$ & $\epsilon$ & $\epsilon$ \\
     \bottomrule
     \end{tabular}
     \caption{An example of bad approximate MMA1 or MMAX guarantee that a PROPX  allocation with symmetric agents provides, where $0 < \epsilon < \frac{1}{6}$.}
     \label{tab: propx_mma1_mmax_example}
     \vspace{-0.5em}
 \end{table}
    
     Now consider an allocation $\mathcal{X} = (X_1, X_2)$ with $X_1 = \{f_1,f_2\}$ and $X_2 = \{f_3, f_4\}$ and assume that this allocation is PROPX to agent 2. \footnote{For agent 2, it is easy to assign costs to the items to make allocation $\mathcal{X}$ be PROPX to her. Therefore, we do not specify the cost of each item from her perspective. The same applies to the subsequent examples.} For agent 1, this is a PROPX allocation, and this allocation is not better than ($\frac{1}{4\epsilon}-\frac{1}{2}$)-MMA1 or ($\frac{1}{4\epsilon}-\frac{1}{2}$)-MMAX since $\frac{c_1(X_1 \setminus \{f\})}{\MMS_1(X_{-1},1)}=\frac{\frac{1}{2}-\epsilon}{2\epsilon}$ for any item $f \in X_{1}$. 
    For any given $\alpha\ge 1$, setting $\epsilon<\frac{1}{4\alpha+2}$, we have $\frac{1}{4\epsilon}-\frac{1}{2} > \alpha$. Thus, this allocation is not $\alpha$-MMA1 or $\alpha$-MMAX.
\end{proof}

The above result directly implies the following corollary.
\begin{corollary}
\label{cor:prop1_mma1__mmax_chores}
A PROP1 allocation is not necessarily $\alpha$-MMA1 or $\alpha$-MMAX for any $\alpha \in [1,+\infty)$.    
\end{corollary}


Similarly, we have the following counterpart results.

\begin{proposition}
\label{pro:mmax_prop1__propx_chores}
    An MMAX allocation may not be PROP1 or PROPX, even when the agents have the same weight.
\end{proposition}
\begin{proof}
    Consider an instance with four agents with identical weights and five items. We focus on agent $1$, and the cost of each item, according to agent 1, is shown in Table \ref{tab: mmax_prop1_propx_example}.
    
    \begin{table}[h]
     \centering
     \begin{tabular}{c|cccccc}
     \toprule
    $f_j$ & $f_1$ & $f_2$ & $f_3$ & $f_4$ & $f_5$ \\
    \midrule
     $c_1(f_j)$ & $\frac{1}{3}$ & $\frac{1}{3}$ & $\frac{1}{3}$ & $0$ & $0$\\
     \bottomrule
     \end{tabular}
     \caption{An example that MMAX cannot imply PROP1 or PROPX.}
     \label{tab: mmax_prop1_propx_example}
     \vspace{-0.5em}
 \end{table}
 
Now consider an allocation $\mathcal{X} = (X_1, X_2, X_3, X_4)$ with $X_1 = \{f_1,f_2\}$, $X_2 = \{f_3\}$, $X_3 = \{f_4\}$, and $X_4 = \{f_5\}$. It is trivial that $X_2$ is MMAX to agent 2, $X_3$ is MMAX to agent 3 and $X_4$ is MMAX to agent 4. For agent 1, $X_1$ is MMAX to her, but this allocation is not PROP1 or PROPX to her.  
\end{proof}

The above result directly implies the following corollary.

\begin{corollary}
    An MMA1 allocation may not be PROP1 or PROPX.
\end{corollary}

When agents have identical weights, if an allocation $\mathcal{X} = (X_1, \ldots,X_n)$ is MMA1 or MMAX, it is obvious that for each agent $i \in N$, we have $c_i(X_i \setminus \{f_{max}^{i}\}) < \frac{1}{2}$ or $c_i(X_i \setminus \{f_{min}^{i}\}) \leq \frac{1}{2}$, where $f_{max}^{i}$ and $f_{min}^{i}$ are the items with the maximum and minimum cost in $X_i$ respectively. Therefore, it is not hard to see that an MMA1 allocation is also $\frac{n}{2}$-PROP1, and an MMAX allocation is also $\frac{n}{2}$-PROPX.


\subsection{MMS}

Finally, we discuss the relationship between MMS and MMA related notions. 

\begin{proposition}
\label{pro: mms_mma1_mmax_chores}
    When the agents have arbitrary weights, an MMS allocation is not necessarily $\alpha$-MMA1 or $\alpha$-MMAX for any $\alpha \in [1,+\infty)$. 
\end{proposition}
\begin{proof}
     Consider an instance with three agents and four items.  Assume that their weights are $w_1 = \frac{1}{2}$, $w_2 = \frac{1}{4}$ and $w_3 = \frac{1}{4}$ respectively. We focus on agent $1$, and the cost of each item, according to agent 1, is shown in Table \ref{tab: propx_mma1_mmax_example}.    
      \begin{table}[h]
     \centering
     \begin{tabular}{c|cccc}
     \toprule
    $f_j$ & $f_1$ & $f_2$ & $f_3$ & $f_4$\\
    \midrule
     $c_1(f_j)$ & $\frac{1}{2} - \epsilon$ & $\frac{1}{2} - \epsilon$ & $\epsilon$ & $\epsilon$ \\
     \bottomrule
     \end{tabular}
     \caption{An example of bad approximate MMA1 or MMAX guarantee that an MMS allocation provides, where $0 < \epsilon < \frac{1}{6}$.}
     \label{tab: mms_mma1_mmax_example}
     \vspace{-0.8em}
     \end{table} 
     
    Note that $\MMS_1 = 1-2\epsilon$. Now consider an allocation $\mathcal{X} = (X_1, X_2, X_3)$ with $X_1 = \{f_1,f_2\}$, $X_2 = \{f_3\}$ and $X_3 = \{f_4\}$ and assume that this allocation is MMS to agent 2 and agent 3. For agent 1, this is an MMS allocation, and this allocation is no better than ($\frac{1}{4\epsilon}-\frac{1}{2}$)-MMA1 or ($\frac{1}{4\epsilon}-\frac{1}{2}$)-MMAX since $\frac{c_1(X_1 \setminus \{f\})}{\MMS_1(X_{-1},2)}=\frac{\frac{1}{2}-\epsilon}{2\epsilon}$ for any item $f \in X_{1}$. For any given $\alpha \geq 1$, setting $\epsilon < \frac{1}{4\alpha+2}$, we have $\frac{1}{4\epsilon}-\frac{1}{2} > \alpha$. Thus, this allocation is not $\alpha$-MMA1 or $\alpha$-MMAX.   
\end{proof}

\begin{proposition}
\label{pro:mms-mma1}
    When the agents have the same weight, 
    \begin{enumerate}
        \item any MMS allocation is also MMA1;
        \item there exists $\alpha \in (1,+\infty)$ such that an $\alpha$-MMS allocation is not $\beta$-MMA1 for any $\beta \in [1,+\infty)$;
        \item an MMS allocation is not necessarily $\alpha$-MMAX for any $\alpha \in [1,+\infty)$.
    \end{enumerate} 
\end{proposition}
\begin{proof}
For the first statement, let $\mathcal{X}=(X_1,\ldots, X_n)$ be an MMS allocation. Suppose, for the contradiction, that $\mathcal{X}$ is not an MMA1 allocation, i.e., there exists some agent $i \in N$ such that $c_i(X_i \setminus \{f\}) > \MMS_i(X_{-i},n-1)$ holds for any item $f \in X_{i} $. Let $\mathcal{Y}= (Y_j)_{j \neq i}$, where $\max_{j \in N \setminus \{i\}}c_i(Y_j) = \MMS_i(X_{-i},n-1)$, be an $(n-1)$-allocation of $X_{-i}$. In allocation $\mathcal{Y}$, we assume that bundle $Y_{k}$ satisfies $c_i(Y_k) = \MMS_i(X_{-i},n-1)$. Next, if we choose one item $f_{max}^{i}$ from $X_{i}$, where $f_{max}^{i} = \arg \max_{f \in X_{i}}c_i(f)$, and put it in bundle $Y_{k}$, we obtain a new $n$-partition of $M$, i.e., $\mathcal{P}^{new} = \{X_i \setminus \{f_{max}^{i}\}, Y_{k} \cup \{f_{max}^{i}\ \}\} \cup \{Y_{j}\}_{j \neq i,k}$. In $\mathcal{P}^{new}$, the cost of the bundle with the maximum cost is strictly less than $c_{i}(X_i)$, which contradicts the definition of $\MMS_i(M,n)$.

Regarding the second statement, consider an instance with two agents with identical weights and three items. We focus on agent $1$, and the cost of each item, according to agent 1, is shown in Table \ref{tab: mms_mma1_example}.

\begin{table}[h]
     \centering
     \begin{tabular}{c|ccc}
     \toprule
    $f_j$ & $f_1$ & $f_2$ & $f_3$ \\
    \midrule
     $c_1(f_j)$ & $\frac{2}{3+\epsilon}$ & $\frac{1}{3+\epsilon}$ & $ \frac{\epsilon}{3+\epsilon}$  \\
     \bottomrule
     \end{tabular}
     \caption{An example of bad approximate MMA1 guarantee that an approximate MMS allocation with symmetric agents provides, where $0 < \epsilon < 1$.}
     \label{tab: mms_mma1_example}
     \vspace{-0.8em}
     \end{table}  
     
Note that $\MMS_1 = \frac{2}{3+\epsilon}$. Now consider an allocation $\mathcal{X} = (X_1, X_2)$ with $X_1 = \{f_1,f_2\}$ and $X_2 = \{f_3\}$ and assume that this allocation is MMS to agent 2. It is easy to see that, for agent 1, this is a $\frac{3}{2}$-MMS allocation, and this allocation is no better than $\frac{1}{\epsilon}$-MMA1 since $\frac{c_1(X_1 \setminus \{f_1\})}{\MMS_1(X_{-1},1)}=\frac{1}{\epsilon}$. For any $\beta \geq 1$, setting $\epsilon < \frac{1}{\beta}$, we have $\frac{1}{\epsilon} > \beta$. Thus, this allocation is not $\beta$-MMA1.

For the last statement, suppose that we have two agents with identical weights and three items. We focus on agent $1$, and the cost of each item, according to agent 1, is shown in Table~\ref{tab: mms_mmax_example}.

\begin{table}[h]
     \centering
     \begin{tabular}{c|ccc}
     \toprule
    $f_j$ & $f_1$ & $f_2$ & $f_3$ \\
    \midrule
     $c_1(f_j)$ & $\frac{1}{1+\epsilon}$ & $\frac{\epsilon}{1+\epsilon}$ & $ 0$  \\
     \bottomrule
     \end{tabular}
     \caption{An example of bad approximate MMAX guarantee that an MMS allocation with symmetric agents provides, where $0 < \epsilon < 1$.}
     \label{tab: mms_mmax_example}
     \vspace{-0.5em}
     \end{table}  
     
It is not hard to check $\MMS_1 = \frac{1}{1+\epsilon}$. Next, consider an allocation $\mathcal{Y} = (Y_1, Y_2)$ with $Y_1 = \{f_1,f_3\}$ and $Y_2 = \{f_2\}$ and assume that this allocation is MMS to agent 2. For agent 1, this is an MMS allocation, and this allocation is no better than $\frac{1}{\epsilon}$-MMAX since $\frac{c_1(X_1 \setminus \{f_3\})}{\MMS_1(X_{-1},1)}=\frac{1}{\epsilon}$. Similarly, for any $\alpha \geq 1$, setting $\epsilon < \frac{1}{\alpha}$, we have $\frac{1}{\epsilon} > \alpha$. Thus, this allocation is not $\alpha$-MMAX. 
\end{proof}

\begin{proposition}
    An MMA allocation may not be MMS, even when the agents have the same weight.
\end{proposition}

\begin{proof}
    Consider an instance with three agents and nine items from \cite{feige2021tight}. We focus on agent $1$, and the cost of each item, according to agent 1, is shown in Table \ref{tab: mma_mms_example}.

    \begin{table}[h]
     \centering
     
     \begin{tabular}{c|ccccccccc}
     \toprule
    $f_j$ & $f_1$ & $f_2$ & $f_3$ & $f_4$ & $f_5$ & $f_6$ & $f_7$ &$f_8$& $f_9$\\
    \midrule
     $c_1(f_j)$ &  $\frac{6}{129}$&  $\frac{15}{129}$&  $\frac{22}{129}$& $\frac{26}{129}$& $\frac{10}{129}$ & $\frac{7}{129}$ & $\frac{12}{129}$ & $\frac{19}{129}$ & $\frac{12}{129}$\\
     \bottomrule
     \end{tabular}
     
     \vspace{-0.5em}
     
     \caption{An example that MMA can not imply MMS.}
     \label{tab: mma_mms_example}
     
 \end{table}
 
It can be verified that $\MMS_1 =\frac{43}{129}$. Consider an allocation $\mathcal{X}=(X_1, X_2, X_3)$ with $X_1 = \{f_1, f_4, f_7\}$, $X_2 = \{f_2,f_5,f_8\}$ and $X_3 = \{f_3,f_6,f_9\}$. Assume that $X_2$ is MMA to agent 2 and $X_3$ is MMA to agent 3. Allocation $\mathcal{X}$ is not an MMS to agent 1 since $c_1(X_1) = \frac{6+26+12}{129}=\frac{44}{129} > \frac{43}{129}$. However, $X_1$ is MMA to agent 1 because $c_1(X_1) = \MMS_{1}(X_{-1},2) =\frac{44}{129}$.
\end{proof}

The above result directly implies the following corollary.

\begin{corollary}
    An MMA1 or MMAX allocation may not be MMS.
\end{corollary}

Before the end of this section, we use the following example to illustrate that there are some scenarios where finding approximate MMAX allocations is preferred over finding approximate MMS allocations. Consider an instance with three agents whose weights are $w_1=\frac{1}{2}$ and $w_2 = w_3 = \frac{1}{4}$ respectively and five items. 
Assume that they have the same cost function shown in Table \ref{tab: mms_mmax}.
\begin{table}[b]
     \centering
     \begin{tabular}{c|cccccc}
     \toprule
    $f_j$ & $f_1$ & $f_2$ & $f_3$ & $f_4$ & $f_5$ \\
    \midrule
     $c(f_j)$ &  $\frac{3}{4+9\epsilon}$&  $\frac{1}{4+9\epsilon}$&  $\frac{3\epsilon}{4+9\epsilon}$& $\frac{3\epsilon}{4+9\epsilon}$& $\frac{3\epsilon}{4+9\epsilon}$\\
     \bottomrule
     \end{tabular}
    \caption{An example that approximate MMAX is preferred, where $\epsilon$ is arbitrarily small.}
    \label{tab: mms_mmax}
 \end{table}
 
It can be verified that $\MMS_1 = \frac{3}{4+9\epsilon}$ and $\MMS_2 = \MMS_3= \frac{3}{8+18\epsilon}$. Allocating $X_1=\{f_1,f_2\}$ to agent 1 is $\frac{4}{3}$-MMS to her, which is not too bad regarding MMS fairness
but is severely unfair since almost all cost is on agent 1.
Fortunately, any (bounded-approximate) MMAX allocation ensures that one agent can only get one of the items $f_1$ and $f_2$, since we have 
\[
\MMS_1(\{f_3,f_4,f_5\},2)= \frac{12\epsilon}{4+9\epsilon} \ll \min\{c(f_1), c(f_2)\},
\]
and similarly for $i=2,3$,
\[
\MMS_i(\{f_3,f_4,f_5\},2)= \frac{3\epsilon}{4+9\epsilon} \ll \min\{c(f_1), c(f_2)\}.
\]

\section{Computing Approximate MMAX Allocations}
By Corollary \ref{coro:MMA1}, we know that an MMA1 allocation can be computed in polynomial time.
In this section, we design efficient algorithms to compute approximate MMAX allocations for agents with arbitrary weights. 

\subsection{Approximation Algorithm for Arbitrary $n$}

By the definition of MMAX, each agent $i$ will compare $c_i(X_i \setminus \{f\})$ for an item $f\in X_i$ and $\MMS_i(X_{-i},n-1)$, where $\MMS_i(X_{-i},n-1)$ depends on the allocation $\mathcal{X}$ instead of being a prefixed share.
That is, an MMAX allocation, on the one hand, requires that every agent has a low cost on her bundle excluding any item (similar to PROPX), and on the other hand, requires the existence of some agent whose bundle is worse off than hers (similar to EFX) no matter how the items that are not allocated to her are distributed.
Proposition \ref{pro:propx_mma1_mmax_chores} shows that a PROPX allocation does not provide any guarantee for MMAX. 
However, we show that we can turn an arbitrary PROPX allocation into an approximate MMAX one with some modifications. 

Consider the following instance, where we have three agents such that $w_1=\frac{1}{2}$ and $w_2=w_3=\frac{1}{4}$, and four items. We assume that they have the same cost function, and the costs of four items are shown in Table \ref{tab:alg:mmax}.
Consider a PROPX allocation $\mathcal{X} = (X_1, X_2,X_3)$ with $X_1 = \{f_1,f_2\}$, $X_2 = \{f_3\}$ and $X_3 = \{f_4\}$. It is clear that $X_2$ and $X_3$ are PROPX to agents 2 and 3, and $X_1=\{f_1, f_2\}$ is PROPX to agent 1.
However, $X_1$ is far from being MMAX to agent 1 since $f_3$ and $f_4$ have infinitesimally small costs compared to the weights of the other two agents.
 \begin{table}[h]
     \centering
     \begin{tabular}{c|cccc}
     \toprule
    $f_j$ & $f_1$ & $f_2$ & $f_3$ & $f_4$  \\
    \midrule
     $c(f_j)$ & $\frac{1}{2}-\epsilon$ & $\frac{1}{2}-\epsilon$ & $\epsilon$ & $\epsilon$ \\
     \bottomrule
     \end{tabular}
     \caption{An example of hard instance, where $\epsilon$ is arbitrarily small.}
     \label{tab:alg:mmax}
 \end{table}
 
In this example, the PROPX allocation fails to provide a good approximation of MMAX to an agent because the items allocated to an agent are all very heavy compared to those allocated to the other agents.  
Then our algorithm starts with an arbitrary PROPX allocation, and if an agent only obtains a single item, her bundle is trivially MMAX to her.
The bad situation is when an agent $i$ gets at least two items where each item has a cost larger than a fraction (a parameter determined by $n$) of that of the items allocated to the other agents. Next we exchange two arbitrary items in $X_i$ with the two smallest bundles among the other agents in agent $i$'s perspective. 
Note that if $n=2$, we can only exchange one item with one bundle.
We check if the bad situation happens for the agents one by one and return the final allocation. 
The formal description of the algorithm is shown in Algorithm \ref{alg: swap algorithm}.




\begin{algorithm}[t]
\caption{Swap algorithm}
\label{alg: swap algorithm}
\KwIn{A PROPX allocation $\mathcal{X} = (X_1, \ldots, X_n)$}
\KwOut{An approximate MMAX allocation $\mathcal{X}^{*}$}
For every agent $i \in N$, let $f_{min}^{i}$ = $\arg \min_{f \in X_i} c_i(f)$ and $f_{max}^{i}$ = $\arg \max_{f \in X_i} c_i(f)$;

{\color{gray} \% - - - - - - - - - - - - - - - - $n = 2$ - - - - - - - - - - - - - - - - }

\If{$n=2$}{
\For{$i = 1$ \KwTo $n$}{\If{$\lvert X_i \rvert  \geq 2$  \ and \ $c_i(f_{min}^{i}) > \lambda \cdot c_i(X_{-i})$ \label{lin: heavy_bundle_criterion_1}}{$X_i = M \setminus \{f_{min}^{i}\}$  and $ X_{-i} = \{f_{min}^{i}\}$;}}}

{\color{gray} \% - - - - - - - - - - - - - - - - $n \ge 3$ - - - - - - - - - - - - - - - -}

\If{$n\geq 3$}{
\For{$i = 1$ \KwTo $n$}{\If{$\lvert X_i \rvert \geq 2$ \ and \ $c_i(f_{min}^{i}) > \lambda \cdot c_i(X_{-i})$ \label{lin: heavy_bundle_criterion_2}}{Choose two bundles $X_j$ and $X_k$ ($j \neq k$) that have the two smallest costs from $X_{-i}$ according to agent $i$;\\ 
$X_{i} = X_{i} \cup X_{j} \cup X_{k} \setminus \{f_{min}^{i}, f_{max}^{i}\}$, $X_{j} = \{f_{min}^{i}\}$ and $ X_{k} = \{f_{max}^{i}\}$;}}}

\Return $\mathcal{X}^{*}$
\end{algorithm}

Before showing the performance of the algorithm, we first prove the following property for PROPX allocations. 

\begin{lemma}
\label{lem:wpropx_property}
     Given any PROPX allocation $\mathcal{X}$\,$=(X_1,\ldots,X_n)$, the following inequality holds for every agent $i \in N$,
    $$ \frac{c_i(X_i \setminus \{f_{min}^{i}\})}{w_i} \leq \frac{c_i(X_{-i}) + c_i(f_{min}^{i})}{1-w_i},$$ 
    where $f_{min}^{i} = \arg \min_{f \in X_i}c_i(f)$.
\end{lemma}
\begin{proof}
By the definition of PROPX allocations, for any agent $i \in N$, we have
\begin{equation}
\label{eqn: wpropx}
    c_i(X_i \setminus \{f_{min}^{i}\}) \leq w_i \cdot c_i(M).
\end{equation}
Next, we expand $c_i(M)$, i.e., $c_i(M) = c_i(X_i \setminus \{f_{min}^{i}\}) + c_i(X_{-i}) + c_i(f_{min}^{i})$, and then Inequality (\ref{eqn: wpropx}) can be changed into the desired form
    \begin{equation*}     
       \frac{c_i(X_i \setminus \{f_{min}^{i}\})}{w_i}  \leq \frac{c_i(X_{-i}) + c_i(f_{min}^{i})}{1-w_i},
    \end{equation*}  
which completes the proof of the lemma.
\end{proof}


\begin{figure}
    \centering 
    \includegraphics[width=0.4 \textwidth]{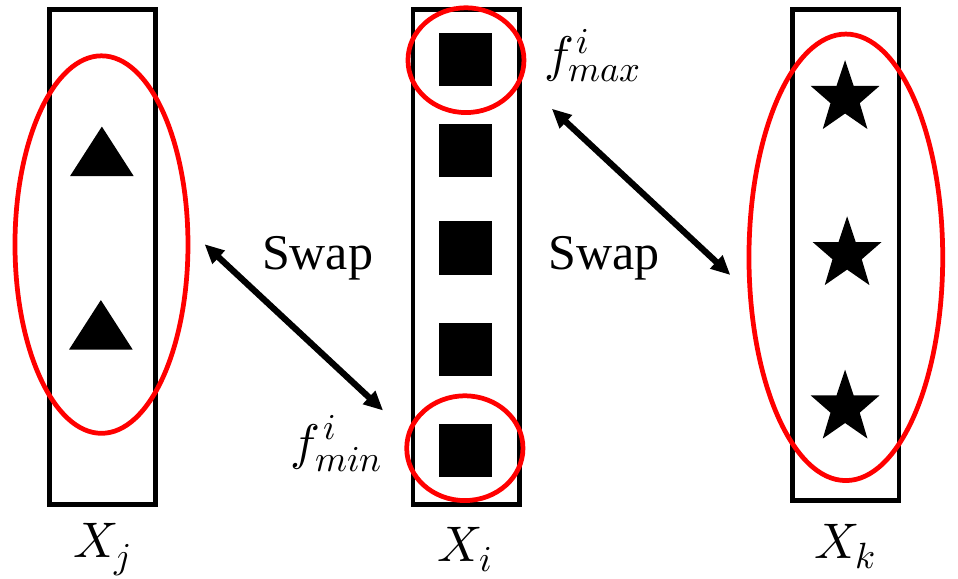}
    \caption{The illustration of \textit{swap}  when $n\geq 3$ and $X_i$ satisfies the condition in Step \ref{lin: heavy_bundle_criterion_2} in Algorithm \ref{alg: swap algorithm}.}
    \label{fig: swap_process}
\end{figure}

\begin{theorem}
\label{thm:wmmax_alg_ratio}
    When $n=2$, Algorithm \ref{alg: swap algorithm} computes a $\frac{3+\sqrt{5}}{2}(\approx 2.618)$-MMAX allocation by setting $\lambda =\frac{1+\sqrt{5}}{2} \approx 1.618$. When $n\ge 3$, Algorithm \ref{alg: swap algorithm} computes a $(1+\lambda)$-MMAX allocation, where $\lambda = \frac{3-n+\sqrt{n^2+10n-7}}{4n-4}$ and $\frac{1}{n-1} < \lambda < \frac{2}{n-1}$.
\end{theorem}
\begin{proof}
PROPX allocations exist and can be computed by the \textit{Bid and Take} algorithm \cite{li2022almost}. Therefore, the input allocation of Algorithm \ref{alg: swap algorithm} can be guaranteed.

Let $\mathcal{X}=(X_1,\ldots,X_n)$ be a PROPX allocation. Fix one agent $i \in N$. If $\lvert X_i \rvert = 1$, $X_i$ is trivially MMAX to agent $i$. 


\underline{Case 1: $n = 2$}. We pick $f_{min}^{i}$ from $X_{i}$, and swap it with $X_{-i}$. Let $X_{i}^{new}$ and $X_{j}^{new}$ denote the new bundles of agent $i$ and the other agent $j \in N \setminus \{i\}$, respectively, after the swap. For these two new bundles, we have $X_{i}^{new} = M \setminus \{f_{min}^{i}\} = X_i \cup X_{-i} \setminus \{f_{min}^{i}\}$ and $X_{j}^{new} = \{f_{min}^{i}\}$. For agent $j$, this new allocation is MMAX to her. 

For agent $i$, we have
\begin{align}   
  \frac{c_i(X_i^{new})}{w_i}  &= \frac{c_i(X_i \setminus \{f_{min}^{i}\}) + c_i(X_{-i})}{w_i} \nonumber
  \\
  & \leq (1+\lambda^{-1}) \cdot \frac{ c_i(X_i \setminus \{f_{min}^{i}\})}{w_i}, \label{eqn:wmmax_phi_first_inequality}
\end{align}
where the inequality follows from $c_i(X_i \setminus \{f_{min}^{i}\}) \geq c_i(f_{min}^{i}) > \lambda \cdot c_i(X_{-i})$, and 
\begin{equation}
  \frac{c_i(X_{-i}) + c_i(f_{min}^{i})}{1-w_i}  \leq (1+\lambda^{-1}) \cdot \frac{ c_i(X_{-i}^{new})}{1-w_i},
\label{eqn:wmmax_phi_second_inequality}
\end{equation}
where $c_i(X_{-i}^{new}) = c_i(f_{min}^{i})$.

Note that $\lambda = \frac{\sqrt{5}+1}{2}$, which is the positive root to the quadratic equation $\lambda^{2}-\lambda-1=0$. Combining Inequality (\ref{eqn:wmmax_phi_first_inequality}) and (\ref{eqn:wmmax_phi_second_inequality}), and by Lemma \ref{lem:wpropx_property}, we get
\begin{equation*}
\frac{c_i(X_i^{new})}{w_i} \leq  (1+\lambda^{-1})^{2} \cdot \frac{ c_i(X_{-i}^{new})}{1-w_i} = (1+\lambda) \cdot \frac{ c_i(X_{-i}^{new})}{1-w_i}.  
\end{equation*}

Therefore, when $n=2$, for any agent $i \in N$, the allocation $\mathcal{X}^{*}$ returned by Algorithm \ref{alg: swap algorithm} is $(1+\lambda)$-MMAX.

\underline{Case 2: $n > 2$}.
Fix one agent $i \in N$. If $\lvert X_i \rvert = 1$, for agent $i$, this allocation is trivially MMAX. If $\lvert X_i \rvert \geq 2$ and $c_i(f_{min}^{i}) \leq \lambda \cdot c_i(X_{-i})$, by Lemma \ref{lem:wpropx_property}, we have
\begin{equation*}
\begin{aligned}
  \frac{c_i(X_i \setminus \{f_{min}^{i}\})}{w_i} & \leq \frac{c_i(X_{-i}) + c_i(f_{min}^{i})}{1-w_i}
  \\
  & \leq (1+\lambda) \cdot \frac{ c_i(X_{-i})}{1-w_i}.    
\end{aligned}  
\end{equation*}

If $\lvert X_i \rvert  \geq 2$  and  $c_i(f_{min}^{i}) > \lambda \cdot c_i(X_{-i})$, we pick $f_{min}^{i}$ and $f_{max}^{i}$ from $X_i$, choose two bundles $X_j$ and $X_k$ ($j \neq k$), where these two bundles have the two smallest costs from $X_{-i}$ in agent $i$'s perspective, and then swap them. Let $X_{i}^{new}$, $X_{j}^{new}$ and $X_{k}^{new}$ denote the new bundles of agents $i$, $j$ and $k$, respectively, after the swap. For these three new bundles, we have $X_{i}^{new} = X_{i} \cup X_{j} \cup X_{k} \setminus \{f_{min}^{i},f_{max}^{i}\}$, $X_j^{new} = \{f_{min}^{i}\}$ and $X_k^{new} = \{f_{max}^{i}\}$. It is easy to see that $X_j^{new}$ and $X_{k}^{new}$ are MMAX to agents $j$ and $k$, respectively.

For agent $i$, we have
\begin{align}
       &  \frac{c_i(X_i^{new})}{w_i} / \frac{c_i(X_i \setminus \{f_{min}^{i}\})}{w_i} \nonumber
    \\
    = \ & \frac{c_i(X_i \setminus \{f_{min}^{i}\}) + c_i(X_j \cup X_k) - c_i(f_{max}^{i})}{c_i(X_{i} \setminus \{f_{min}^{i}\})} \nonumber
    \\
    < \ & \frac{c_i(X_{i} \setminus \{f_{min}^{i}\}) + (\frac{2}{n-1}-\lambda) \cdot c_i(X_{-i})}{c_i(X_{i} \setminus \{f_{min}^{i}\})} \nonumber
    \\
    < \ & \frac{c_i(X_{i} \setminus \{f_{min}^{i}\}) + (\frac{2}{(n-1)\lambda}-1) \cdot c_i(X_{i} \setminus \{f_{min}^{i}\})}{c_i(X_{i} \setminus \{f_{min}^{i}\})} \nonumber
    \\
    = \ & \frac{2}{(n-1)\lambda},\label{eqn:wmmax_x0_first_equality}
\end{align}  
where the first inequality follows from $c_i(X_j \cup X_k) \leq \frac{2}{n-1} \cdot c_i(X_{-i})$ and $c_{i}(f_{max}^{i}) \geq c_i(f_{min}^{i}) > \lambda \cdot c_i(X_{-i}{})$, and the second from $\frac{2}{n-1} > \lambda$, which is explained at the end of the proof, and $c_i(X_i \setminus \{f_{min}^{i}\}) \geq c_i(f_{min}^{i}) > \lambda \cdot c_i(X_{-i})$.

Besides that, we have 
\begin{align}
    & \frac{ c_i(X_{-i}^{new})}{1-w_i} / \frac{c_{i}(X_{-i}) + c_i(f_{min}^{i})}{1-w_i} \nonumber
    \\
    = \ & \frac{c_i(X_{-i}) - c_i(X_j \cup X_k) + c_i(f_{max}^{i}) + c_i(f_{min}^{i})}{c_i(X_{-i}) + c_i(f_{min}^{i})}   \nonumber
    \\
    \geq \ & \frac{\frac{n-3}{n-1}\cdot c_i(X_{-i})  + 2 \cdot c_i(f_{min}^{i})}{c_i(X_{-i}) + c_i(f_{min}^{i})}   \nonumber
    \\
    = \ & \frac{n-3}{n-1} +  \frac{\frac{n+1}{n-1} \cdot c_i(f_{min}^{i})}{c_i(X_{-i})+ c_i(f_{min}^{i})}  \nonumber
    \\
    > \ & \frac{n-3}{n-1} + \frac{\frac{n+1}{n-1}}{\frac{1}{\lambda}+ 1} 
    = \frac{(2n-2)\lambda + (n-3)}{(n-1)(\lambda+1)}, \label{eqn:wmmax_x0_second_equality} 
\end{align}  
where the first inequality follows from $c_i(X_j \cup X_k) \leq \frac{2}{n-1} \cdot c_i(X_{-i})$ and $c_{i}(f_{max}^{i}) \geq c_i(f_{min}^{i})$.

 Note that $\lambda =  \frac{3-n+\sqrt{n^2+10n-7}}{4n-4}$, which is a root to the quadratic equation $(2n-2)\lambda^{2}+(n-3)\lambda-2=0$. Putting Inequalities (\ref{eqn:wmmax_x0_first_equality}) and (\ref{eqn:wmmax_x0_second_equality}) together, and by Lemma \ref{lem:wpropx_property}, we have 
\begin{equation*}
\begin{aligned}
\frac{c_i(X_i^{new})}{w_i} & < \frac{2}{(n-1)\lambda} \cdot \frac{c_i(X_i \setminus \{f_{min}^{i}\})}{w_i} 
  \\
   & < \frac{2}{(n-1)\lambda} \cdot \frac{c_i(X_{-i})+c_i(f_{min}^{i})}{1-w_i} 
  \\
  & <   \frac{2\lambda+2}{(2n-2)\lambda^{2} + (n-3)\lambda} \cdot \frac{c_i(X_{-i}^{new})}{1-w_i} 
  \\
  & = (1+\lambda) \cdot \frac{ c_i(X_{-i}^{new})}{1-w_i}.    
\end{aligned} 
\end{equation*}

Therefore, $X_i^{new}$ is $(1+\lambda)$-MMAX to agent $i$.

Regarding the range of $\lambda$, we have
\begin{equation*}
\begin{aligned}
    \lambda - \frac{2}{n-1} = \frac{\sqrt{n^2+10n-7}-n-5}{4n-4} < 0;
    \\
    \lambda - \frac{1}{n-1} = \frac{\sqrt{n^2+10n-7}-n-1}{4n-4} > 0.
\end{aligned}
\end{equation*}

Overall, the allocation $\mathcal{X}^{*}$ returned by Algorithm \ref{alg: swap algorithm} is $(1+\lambda)$-MMAX. 
\end{proof}

By \citet{hajiaghayi2023almost}, the lower bound of EFX allocations for two agents with different weights is 1.272. The lower bounds for the case when $n\ge 3$ is unknown.

\subsection{Improved Approximation for $n=2$}

Recall the result of Algorithm \ref{alg: swap algorithm}. When there are two agents, it computes a $\frac{3+\sqrt{5}}{2}$-MMAX allocation, which has a big gap compared to the results ($<$1.8) when there are more than two agents. Therefore, we want to design a new algorithm to improve it. In Algorithm \ref{alg: swap algorithm}, \textit{swap} only happens when for one agent, the smallest cost item is greater than a fraction of the cost of items not allocated to her. Now consider the following instance with two agents and four items. Assume that their weights are $w_1 = \frac{3-\sqrt{5}}{2}$ and $w_2 = \frac{\sqrt{5}-1}{2}$, and they have the same cost function.

\begin{table}[h]
     \centering
     \begin{tabular}{c|cccc}
     \toprule
    $f_j$ & $f_1$ & $f_2$ & $f_3$ & $f_4$  \\
    \midrule
     $c(f_j)$ & $\frac{3-\sqrt{5}}{2}$ & $\frac{3-\sqrt{5}}{2}$ & $\sqrt{5}-2$& 0 \\  
     \bottomrule
     \end{tabular}
 \end{table}

Consider a PROPX allocation $\mathcal{X}=(X_1,X_2)$ with $X_1=\{f_1,f_2\}$ and $X_2 =\{f_3,f_4\}$. If Algorithm \ref{alg: swap algorithm} is implemented with the input function $\mathcal{X}$, \textit{swap} will not happen since it does not meet the condition of \textit{swap}. However, if we directly transfer $f_2$ from $X_1$ to $X_2$, we obtain an MMAX allocation. In this example, we find that when the cost of the smallest item is not super large, we may still make some changes between the bundles of two agents to obtain a better approximate MMAX allocation.

Next, we show how to improve the approximation for $n=2$.
Given any ordered instance, we first allocate each item, from largest to smallest, to the agent with the lower cost.
If we do not meet any item whose addition exceeds an agent's weight, the allocation is naturally PROP and thus MMA.
If not, consider the first time when the item, say $f_j$ being allocated to agent $i$, the cost of agent $i$'s bundle exceeds $w_i$.
Note that before allocating $f_j$, the partial allocation is still PROP, and thus the intuition is as follows. 
If $f_j$ is super large, we can give $f_j$ -- a single item -- to one agent, and all the other items are given to the other agent;
otherwise, we can allocate  $f_j$ to the agent that results in a better approximation ratio.
The formal description of the algorithm is shown in Algorithm \ref{alg: draft_swap algorithm}.


\begin{algorithm}[h]
\caption{Improved algorithm for two agents}
\label{alg: draft_swap algorithm}
\KwIn{An ordered instance $\mathcal{I}=\langle M,N,\boldsymbol{c},\boldsymbol{w} \rangle$ where $\lvert N \rvert = 2$}
\KwOut{An approximate MMAX allocation $\mathcal{X}$}

Let $\mathcal{X}=(\emptyset, \emptyset)$;

\For{$j=1$ \KwTo $m$}{\If{$c_1(f_j) \leq c_2(f_j)$ and $c_1(X_1) \leq w_1$}{$X_1 = X_1 \cup \{f_j\}$;}
\ElseIf{$c_1(f_j) \geq c_2(f_j)$ and $c_2(X_2) \leq w_2$}{$X_2 = X_2 \cup \{f_j\}$;}
\Else{

\label{step:n=2:cases}Let $i=1$ if $c_1(f_j) \leq c_2(f_j)$ and $i = 2$ otherwise; 

Let $l = 3 - i$; 

Let $R = M\setminus (X_1 \cup X_2 \cup \{f_j\})$;




\If{$\lvert X_i \rvert \geq 1$ and $c_i(f_{j}) > \mu \cdot c_i(X_l \cup R)$}{\label{step:n=2:cases:cases:1}$X_i = M\setminus \{f_j\}$ and $X_l = \{f_{j}\}$;}
\ElseIf{$w_i \cdot c_l(M\setminus X_i) \leq \lambda \cdot w_l \cdot c_l(X_i)$}{\label{step:n=2:cases:cases:2}$X_i = X_i$ and $X_l = M \setminus X_i$;}
\Else{\label{step:n=2:cases:cases:3}$X_i = X_i \cup \{f_j\}$ and $X_{l} = X_{l} \cup R$;}


\label{step:n=2:cases:cases:end}Break the \textit{for} loop;}}






\Return $\mathcal{X}$

\end{algorithm}

\begin{theorem}
    Algorithm \ref{alg: draft_swap algorithm} computes a $1.91$-MMAX allocation by setting $\lambda =1.91$ and $\mu=2.63$.
\end{theorem}

\begin{proof}
If Steps \ref{step:n=2:cases} to \ref{step:n=2:cases:cases:end} do not happen, then we have $c_i(X_i) \leq w_i$ for both agents $i \in \{1,2\}$, which implies this allocation is PROP. By Lemma \ref{lem: ef_prop_mma_chores}, this allocation is also MMA.

Otherwise, suppose when we allocate item $f_j$ to agent $i$, the \textit{for} loop is broken. Recall that $l = 3 - i$ is the other agent and $R = M\setminus (X_1 \cup X_2 \cup \{f_j\})$ is the set of unconsidered items so far. Then we have three cases happen at Steps \ref{step:n=2:cases:cases:1} (Case 1), \ref{step:n=2:cases:cases:2} (Case 2) and \ref{step:n=2:cases:cases:3} (Case 3).

Case 1: $\lvert X_i \rvert \geq 1$ and $c_i(f_j) > \mu \cdot c_i(X_l \cup R)$, when the algorithm allocates $f_j$ to agent $l$ and $M \setminus \{f_j\}$ to agent $i$. 
Denote this allocation by $X_i^{new} = M\setminus \{f_j\}$ and $X_l^{new} = \{f_{j}\}$.
It is easy to see that $X_l^{new}$ is MMAX to agent $l$. 
For agent $i$, we have 
\begin{align}   
  \frac{c_i(X_i^{new})}{w_i}  &= \frac{c_i(X_i) + c_i(X_l \cup R)}{w_i} \nonumber
  \\
  & < (1+\frac{1}{\mu}) \cdot \frac{ c_i(X_i)}{w_i}, \label{eqn: wmmax_mu_first_inequality}
\end{align}
where the inequality follows from $c_i(X_i) \geq c_i(f_j) > \mu \cdot c_i(X_l \cup R)$. 





Note that before allocating item $f_j$, we have $c_i(X_i) \leq w_i$. Because $c_i(M) = c_i(X_i) + c_i(X_l \cup R) + c_i(f_j) = 1$, we can rewrite $c_i(X_i) \leq w_i$ into the following form
\begin{equation}
\label{eqn: wmmax_mu_prop}
    \frac{c_i(X_i)}{w_i} \leq \frac{c_i(X_l \cup R) + c_i(f_j)}{1-w_i}.
\end{equation}

Combining Inequalities (\ref{eqn: wmmax_mu_first_inequality}) and (\ref{eqn: wmmax_mu_prop}), we get 
\begin{equation*}
\begin{aligned}
    \frac{c_i(X_i^{new})}{w_i} &< (1+\frac{1}{\mu}) \cdot \frac{c_i(X_l \cup R) + c_i(f_j)}{1-w_i} \\
 &< (1+\frac{1}{\mu})^{2} \cdot \frac{ c_i(X_{l}^{new})}{1-w_i} < 1.91\cdot \frac{ c_i(X_{l}^{new})}{1-w_i},
\end{aligned}
\end{equation*}
where the second inequality follows from $c_i(X_l \cup R) < \frac{1}{\mu} \cdot c_i(f_j)$ and $X_l^{new}=\{f_j\}$, and the last equality is because of $\mu = 2.63$, and thus $X_i^{new}$ is 1.91-MMAX to agent $i$.

Case 2: $w_i \cdot c_l(M \setminus X_i)$ $ \leq \lambda \cdot w_l \cdot c_l(X_i)$, given $c_i(f_j) \leq \mu \cdot c_i(X_{l} \cup R)$, and the algorithm allocates $f_j$ and $R$ to agent $l$ and agent $i$ keeps her current bundle. 
Note that when Case 2 happens, we also have $\lvert X_i \rvert \geq 1$.
Again, we denote this allocation by $X_i^{new}= X_i$ and $X_l^{new}= M \setminus X_i$.
Since $c_i(X_i) \leq w_i$, the allocation is PROP and thus MMA to agent $i$. 
For agent $l$, we have
\begin{equation*}
\begin{aligned}
    \frac{c_l(X_l^{new})}{w_l}\leq \lambda \cdot \frac{c_l(X_i^{new})}{w_i} = \lambda \cdot \frac{c_l(X_i^{new})}{1-w_l},
\end{aligned}
    \end{equation*}
and since $\lambda = 1.91$, $X_{l}^{new}$ is $1.91$-MMAX to agent $l$.

Case 3: $\lvert X_i \rvert \geq 1$, $c_i(f_j) \leq \mu \cdot c_i(X_{l} \cup R)$ and $w_i \cdot c_l(M \setminus X_i) > \lambda \cdot w_l \cdot c_l(X_i)$ or $X_i = \emptyset$, when the algorithm continues to allocate $f_j$ to agent $i$ and $R$ to agent $l$. 
Let $X_i^{new} = X_i \cup \{f_j\}$ and $X_l^{new} = X_l \cup R$ be the final allocation. 
By the design of the algorithm, $c_l(X_i \cup \{f_j\}) \geq c_i(X_i \cup \{f_j\}) > w_i= 1-w_l$, implying $c_l(X_l \cup R) < w_l$, and thus the allocation is PROP and MMA to agent $l$.
For agent $i$, we have
\begin{equation*}
\begin{aligned}
  & \frac{c_i(X_i^{new} \setminus \{f_j\})}{w_i} \\
   = &\frac{c_i(X_i)}{w_i} \leq \frac{c_l(X_i)}{w_i}
   <  \frac{1}{\lambda} \cdot \frac{c_l(M \setminus X_i)}{w_l}\\
   \leq &\frac{1}{\lambda} \cdot \frac{c_i(M \setminus X_i)}{w_l}
    =  \frac{1}{\lambda} \cdot  \frac{c_i(X_l \cup R) + c_i(f_j)}{w_l} 
  \\
  \leq& \frac{1+\mu}{\lambda} \cdot \frac{c_i(X_l \cup R)}{1-w_i}
  \le  1.91 \cdot \frac{c_i(X_{l}^{new})}{1-w_i},  
\end{aligned}
 \end{equation*}
where the last inequality is because of $\lambda =1.91$ and $\mu=2.63$, and thus $X_{i}^{new}$ is $1.91$-MMAX to agent $i$.

Combining all three cases, the theorem is proved. 
\end{proof}









\paragraph{Remark.} When there are two agents, MMAX is equivalent to EFX. Therefore, Algorithm \ref{alg: draft_swap algorithm} also ensures a 1.91-EFX allocation for two agents.

\section{Conclusion}
We study the MMA allocation of indivisible chores, when the agents can be asymmetric. 
In general, MMA related notions (including MMA1 and MMAX) are weaker than the counterpart notions of EF but not comparable with those of PROP.
The positive message from this work includes the following. 
MMA1 allocations always exist for agents with arbitrary weights.
MMAX allocations exist when the agents have the same weight, and admit good approximation when the weights are arbitrary, where the approximation ratio gets close to 1 when the number of agents is large.
An interesting future direction is to improve the upper and lower bounds for the approximation ratio of MMAX allocations.

\bibliographystyle{plainnat}
\bibliography{reference}

\clearpage
\appendix

\section{Generalization and Improvements for Goods}
\label{sec: improvements_goods}

\subsection{Definitions and notions}

Let $N=\{1,\ldots,n\}$ be a set of $n$ agents, and $M = \{e_1,\ldots,e_m\}$ be a set of $m$ indivisible goods. Each agent $i \in N$ has a valuation function $v_{i}: 2^{N} \rightarrow \mathbb{R}_{\geq 0}$. The valuation functions are assumed to be additive, i.e., for any $S\subseteq M$, $v_i(S) = \sum_{e\in S} v_i(\{e\})$. 
For simplicity, we use $v_i(e)$ instead of $v_i(\{e\})$ for $e \in M$. 
We study the weighted setting, where each agent has a weight $w_i > 0$, and the weights add up to one, i.e., $\sum_{i \in N}w_i =1$. 
A goods allocation instance is denoted as $\mathcal{I} = \langle M, N, \boldsymbol{v}, \boldsymbol{w}\rangle$, where $\boldsymbol{v} = (v_1,\ldots,v_n)$ and $\boldsymbol{w}= (w_1,\ldots,w_n)$. 

Let $\Pi_{n}(M)$ denote the set of all $n$-partitions of the items. An allocation $\mathcal{A} = (A_1, \ldots, A_n) \in \Pi_{n}(M)$, where $A_i$ is the bundle allocated to agent $i$, is the partition of the set of $m$ items among $n$ agents, i.e., $\cup_{i \in N}A_{i} = M$ and $A_{i} \cap A_{j} = \emptyset$ for any two agents $i \neq j$. Particularly, we denote the set of items that are not allocated to agent $i$ as $A_{-i}$, where $A_{-i} = \cup_{j \in N \setminus \{i\}}A_j$.

Given an instance $\mathcal{I} = \langle M, N, \boldsymbol{v}, \boldsymbol{w} \rangle$, the \textit{maximin share} of agent $i$ on $M$ among $n$ agents is defined as :
$$ \MMS_{i}(M, n) =  w_i \cdot \max_{ \mathcal{B} \in \Pi_{n}(M)} \min_{j \in N} \frac{ v_{i}(B_j)}{w_j}.$$

\begin{definition}[\textbf{MMS}]
For any $\alpha \in [0,1]$, an allocation $\mathcal{A}=(A_1,\ldots,A_n)$ is $\alpha$-approximate maximin share fair ($\alpha$-MMS), if for any agent $i \in N$, $v_i(A_i) \geq \alpha \cdot \MMS_{i}(M, n).$
\end{definition}

\begin{definition}[\textbf{MMA}]

For any $\alpha \in [0,1]$, an allocation $\mathcal{A}=(A_1,\ldots,A_n)$ is $\alpha$-approximate maximin aware fair ($\alpha$-MMA), if for any agent $i \in N$, 
\[
v_i(A_i) \geq \alpha \cdot \MMS_{i}(A_{-i}, n-1),
\]
where  
\[
\MMS_{i}(A_{-i}, n-1) = w_i \cdot \max_{\mathcal{C} \in  \Pi_{n-1}(A_{-i})} \min_{j \in N \setminus \{i\}} \frac{v_{i}(C_j)}{w_j}.
\]
\end{definition}

\begin{definition}[\textbf{MMA1 and MMAX}]

For any $\alpha \in [0,1]$, an allocation $\mathcal{A}=(A_1,\ldots,A_n)$ is $\alpha$-approximate maximin aware up to one item ($\alpha$-MMA1), if for any agent $i \in N$, there exists one item $e \in A_{-i}$ such that 
\begin{equation}
\label{eq:mma1_good}
    v_i(A_i) \geq \alpha \cdot \MMS_{i}(A_{-i} \setminus \{e\}, n-1).
\end{equation}
Similarly, the allocation is $\alpha$-approximate maximin aware up to any item ($\alpha$-MMAX) if Inequality (\ref{eq:mma1_good}) holds for any item $e \in A_{-i}$.
\end{definition}

\begin{definition}[\textbf{Ordered instance}]
\label{def: ordered_instance_goods}

Given an instance $\mathcal{I} = \langle M, N, \boldsymbol{v}, \boldsymbol{w} \rangle$, the corresponding ordered instance of $\mathcal{I}$ is denoted as $\mathcal{I}' = \langle M, N, \boldsymbol{v'}, \boldsymbol{w} \rangle$, where for each agent $i \in N$ and item $e_{j} \in M$, we have $v_i^{'}(e_j) = v_i(e^{j})$ ($e^{j}$ is the $j^{th}$ highest value item of $M$ for agent $i$).
\end{definition}

The reduction can also work in MMA1 and MMA1 allocation for the case of goods, and the proof is similar to that of the chores. Here we directly give the following two lemmas.

\begin{lemma}
\label{lem: wmma1_reduction_approximation_goods}
For any $\alpha \in [0,1]$, suppose that there is a polynomial time algorithm that computes an $\alpha$-MMA1 allocation $\mathcal{A}^{'}=(A_1^{'}, \ldots, A_n^{'})$ of an ordered instance $\mathcal{I}^{'}$, then we have a polynomial time algorithm that computes an $\alpha$-MMA1 allocation $\mathcal{A} = (A_1, \ldots, A_n)$ of instance $\mathcal{I}$.
\end{lemma}


\begin{lemma}
\label{lem: wmmax_reduction_approximation_goods}
For any $\alpha \in [0,1]$, suppose that there is a polynomial time algorithm that computes an $\alpha$-MMAX allocation $\mathcal{A}^{'}=(A_1^{'}, \ldots, A_n^{'})$ of an ordered instance $\mathcal{I}^{'}$, then we have a polynomial time algorithm that computes an $\alpha$-MMAX allocation $\mathcal{A} = (A_1, \ldots, A_n)$ of instance $\mathcal{I}$.
\end{lemma}

\subsection{New findings for goods}

\begin{proposition}
\label{pro: mms_mma1_goods}
    For any $\alpha \in [0,1]$, any $\alpha$-MMS allocation is also $\frac{\alpha}{1+\rho(1-\alpha)}$-MMA1, where $\rho = \max_{i,j \in N} \frac{w_i}{w_j}$.
\end{proposition}

\begin{proof}
Let $\mathcal{A} = (A_1, \ldots, A_n)$ be an $\alpha$-MMS allocation. Suppose, towards a contradiction, that $\mathcal{A}$ is not an $\frac{\alpha}{1+ \rho (1-\alpha)}$-MMA1 allocation, i.e, there exists some agent $i \in N$ such that $v_i(A_i) < \frac{\alpha}{1+\rho(1-\alpha)} \cdot \MMS_i(A_{-i} \setminus \{e\},n-1)$ holds for any item $e \in A_{-i}$. Let $e_{max}^{i}$ denote the item with the maximum value in $A_{-i}$. It is not hard to see that $\MMS_{i}(A_{-i} \setminus \{e_{max}^{i}\},n-1) \leq \MMS_{i}(A_{-i} \setminus \{e\},n-1)$ for any item $e \in A_{-i}$. Therefore, it suffices to use $\MMS_{i}(A_{-i} \setminus \{e_{max}^{i}\},n-1)$ to complete the following proof.

For simplicity, we assume that $\MMS_{i}(M,n) = 1$. So we have $v_i(A_i) \geq \alpha$ and $\MMS_{i}(A_{-i} \setminus \{e_{max}^{i}\},n-1) > 1+\rho(1-\alpha)$. Let $\mathcal{B} = (B_j)_{j \neq i} $, where $w_i \cdot \min_{j \in N \setminus \{i\}} \frac{v_i(B_j)}{w_j} = \MMS_{i}(A_{-i} \setminus \{e_{max}^{i}\},n-1)$, denote an $(n-1)$-allocation of $A_{-i} \setminus \{e_{max}^{i}\}$. Thus, for any bundle $B_{j} \in \mathcal{B}$, we have $v_i(B_j) \geq \frac{w_j}{w_i} \cdot \MMS_{i}(A_{-i} \setminus \{e_{max}^{i}\},n-1) > \frac{w_j}{w_i}[1+\rho(1-\alpha)].$ Now, consider the following two cases:

\underline{Case 1: $v_1(e_{max}^{i}) > 1 - \alpha$.} We can put $e_{max}^{i}$ into $A_i$, and then obtain a new allocation of $M$, i.e., $\mathcal{C} = (A_i \cup \{e_{max}^{i}\}, B_{j})_{(j \neq i)}$. In allocation $\mathcal{C}$, we have $\min \{v_i(A_i \cup \{e_{max}^{i}\}), w_i \cdot \min_{j \in N \setminus \{i\}}\frac{v_i(B_j)}{w_j} \} > 1$, which contradicts the definition of $\MMS_{i}(M,n)$.

\underline{Case 2: $v_1(e_{max}^{i}) \leq 1 - \alpha$.} Let $v_i(e_{max}^{i}) = 1 - \alpha -\epsilon$, where $ 0 \leq \epsilon < 1-\alpha$. Firstly, we also put $e_{max}^{i}$ into $A_i$. After that, we arbitrarily choose one bundle $B_{k}$ which includes more than two items from $\{B_{j}\}_{j \neq i}$. Assume that all items $\{e_{(l)}\}$ in $B_{k}$, where $l = 1, \ldots,  \lvert B_{k} \rvert$, are sorted in the increasing order of value. Next, we remove items from $B_{k}$ in the ascending order of index successively and stop removing when the value of all removed items starts to be strictly greater than $\frac{w_k}{w_i}$, e.g., when item $e_{(t)}$, where $1 \leq t < \lvert B_k \rvert$, is removed, we have $\sum_{l = 1}^{t-1} v_i(e_{(l)}) \leq \frac{w_k}{w_i}$ and $\sum_{l = 1}^{t} v_i(e_{(l)}) > \frac{w_k}{w_i}$, and then we stop removing. Assume that $e_{(t)}$ is the last removed item. Because $v_i(e_{max}^{i})  = 1 - \alpha - \epsilon$, we have $v_i(e_{(l)}) \leq 1-\alpha-\epsilon$ for any $l = 1, \ldots, \lvert B_{k} \rvert$. Thus, we have $\sum_{l=1}^{t}v_i(e_{(l)}) = \sum_{l=1}^{t-1}v_i(e_{(l)})+ v_i(e_{t}) \leq \frac{w_k}{w_i} + (1-\alpha - \epsilon)$ and $\sum_{l=t+1}^{\lvert B_k \rvert}v_i(e_{(l)}) = v_i(B_{k}) - \sum_{l=1}^{t}v_i(e_{(l)}) > (\rho \cdot \frac{w_k}{w_i} -1) (1-\alpha)+\epsilon$. Note that $\rho = \max_{i,j \in N}\frac{w_i}{w_j}$, we have $\rho \cdot \frac{w_k}{w_i} \geq 1$, which implies $\sum_{l=t+1}^{\lvert B_{k} \rvert}v_i(e_{(l)}) > \epsilon$. At last, we put the removed $t$ items back to $B_k$ and the remaining $\lvert B_k \rvert - t$ items into $A_i$, and then obtain a new allocation of $M$, i.e., $\mathcal{D} = (A_i \cup \{e_{max}^{i}\} \cup \{e_{(t+1)}, \ldots, e_{(\lvert B_{k} \rvert)}\}, B_{j})_{j \neq i}$. In allocation $\mathcal{D}$, we have $\min \{v_i(A_i \cup \{e_{max}^{i}\} \cup \{e_{(t+1)}, \ldots, e_{(\lvert B_{k} \rvert)}\}), w_i \cdot \min_{j \in N \setminus \{i\}}\frac{v_i(B_j)}{w_j}\} > 1$, which also contradicts the definition of $\MMS_i(M,n)$.
\end{proof}

The above result directly implies the following corollary.
\begin{corollary}
\label{cor: mms_mma1_goods}
When agents have identical weights, for any $\alpha \in [0,1]$, any $\alpha$-MMS allocation is also $\frac{\alpha}{2-\alpha}$-MMA1.
\end{corollary}
The best known approximation ratio for MMS allocations with symmetric agents is $(\frac{3}{4}+\frac{3}{3836})$ \cite{akrami2023breaking}. By Corollary \ref{cor: mms_mma1_goods}, a $\frac{3}{4}$-MMS allocation is also $\frac{3}{5}$-MMA1, which fixes a misstatement and improves the result of $\frac{1}{2}$-MMA1 in the work of \cite{chan2019maximin}.

\begin{proposition}
    When agents have arbitrary weights, MMA1 and MMAX allocations may not exist. 
\end{proposition}

\begin{proof}

We construct the following instance where no MMA1 and MMAX allocation can be found. The intuition of this instance is that we allow the weight of one agent to be extremely large and keep the weights of the remaining agents small.

Let us consider an instance with $n$ agents and $3n-1$ items where $n \geq 4$ and let $w_i = \epsilon$ ($i = 1,\ldots,n-1$) and $w_n = 1-(n-1)\epsilon$ where $0 < \epsilon < 1$. Assume that the valuation functions of the first $n-1$ agents are the same, and we have
\begin{equation*}
    v_i(e_j) = 
    \begin{cases}
         \frac{\epsilon}{2} & \ if \ j \leq 2n-2
        \\
         \frac{1-(n-1)\epsilon}{2}& \ if \ 2n-1 \leq j \leq 2n
        \\
        0 & \ if \ j \geq 2n+1
        \\      
    \end{cases}
\end{equation*}

For the last agent, we have
\begin{equation*}
    v_n(e_j) = 
    \begin{cases}
       \frac{1-(n-1)\epsilon}{2n} & \ if \ j \leq 2n
       \\
       \epsilon
       & \ if \ j \geq 2n +1
       \\
    \end{cases}
\end{equation*}

Note that $\MMS_i = \epsilon$ ($i=1,\ldots,n-1$) and $\MMS_n = 1-(n-1)\epsilon$. In any allocation that guarantees an MMAX and MMA1 allocation where the value of every bundle is positive, at most $n+1$ items from the first $2n$ items are allocated to agent $n$ since the first $n-1$ agents have zero valuation for the last $n-1$ items. 

Suppose that $\mathcal{A} = (A_1, \ldots, A_n)$ is an MMA1 or MMAX allocation where the value of every bundle is positive. Here we choose sufficiently small $\epsilon$ which satisfies $1-(n-1)\epsilon \geq 2\phi n\epsilon$ where $\phi \approx 1.618$ is the \textit{golden ratio}. Consider the following two cases:

(1) If $\lvert A_n \rvert = 2n$, the value of agent $n$'s bundle is maximized since she cannot pick an extra item, and it is not hard to see that this allocation is MMA1 or MMAX to her since the first $n-1$ agents have exactly one item in their own bundles, respectively. For the first $n-1$ agents, If the value of one agent's bundle is $\frac{1-(n-1)\epsilon}{2}$ that is greater than $\epsilon$, it is not hard to check this allocation is MMA1 or MMAX to her. If the value of one agent's bundle is $\frac{\epsilon}{2}$. For this agent, her MMA1 approximation ratio is
\begin{equation}
\label{wmma1_ratio_computation}
    \frac{\frac{\epsilon}{2}}{ \frac{\epsilon}{1-(n-1)\epsilon} \cdot (\frac{1-(n-1)\epsilon}{2}+ \frac{\epsilon}{2})} = \frac{1}{1+\frac{\epsilon}{1-(n-1)\epsilon}}  < 1.
\end{equation}
and MMAX approximation ratio is
\begin{equation}
\label{wmmax_ratio_computation}
    \frac{\frac{\epsilon}{2}}{ \frac{\epsilon}{1-(n-1)\epsilon} \cdot \frac{1-(n-1)\epsilon}{2} \cdot 2} = \frac{1}{2}.
\end{equation}

Therefore, in this case, this allocation is not MMA1 and MMAX to some agents from the first $n-1$ agents. 

(2) If $\lvert A_n \rvert < 2n$, for agent $n$, if she has more items, the value of her bundle will increase, but $\MMS_n(A_{-n} \setminus \{e\}, n-1)$($e \in A_{-n}$) will be nonincreasing. Therefore, her maximum MMA1 and MMAX approximation ratio are achieved when agent $n$ has $2n-1$ items that include $n+1$ items from the first $2n$ items and $n-2$ items from the remaining $n-1$ items. Thus, agent $n$'s MMA1 approximation ratio is
    \begin{align*}
    & \frac{\frac{1-(n-1)\epsilon}{2n}\cdot (n+1) + (n-2)\epsilon}{\frac{1-(n-1)\epsilon}{\epsilon} \cdot \epsilon} 
     \\
       =  \ & \frac{\frac{n+1}{2n}+\epsilon(n-2-\frac{n^2-1}{2n})}{1-(n-1) \epsilon} 
     \\
      \leq \ & \frac{\frac{n+1}{2n}+\frac{1}{(2\phi+1)n-1}\cdot(n-2-\frac{n^2-1}{2n})}{1-(n-1) \cdot \frac{1}{(2\phi+1)n-1}} 
     \\
      =  \ & \frac{(\phi+1)n + (\phi-2)}{2\phi n} 
     \\
      = & \frac{\phi+1}{2\phi} -\frac{(2-\phi)}{2\phi n} < 1,
    \end{align*}
where the first inequality follows from $n-2-\frac{n^2-1}{2n}>0$ and $\epsilon \leq \frac{1}{(2\phi+1)n-1}$, and her MMAX approximation ratio is
\begin{equation*}
    \begin{split}
      & \frac{\frac{1-(n-1)\epsilon}{2n}\cdot (n+1) + (n-2)\epsilon}{\frac{1-(n-1)\epsilon}{\epsilon} \cdot \frac{1-(n-1)\epsilon}{2n}} 
      \\
       = \ &  \frac{\frac{n+1}{2n}+\epsilon(n-2-\frac{n^2-1}{2n})}{\frac{(n-1)^2 \epsilon}{2n}+\frac{1}{2n\epsilon}-\frac{n-1}{n}} \\
       \leq \ & \frac{\frac{n+1}{2n}+\frac{1}{(2\phi+1)n-1}\cdot(n-2-\frac{n^2-1}{2n})}{\frac{(n-1)^2}{2n} \cdot \frac{1}{(2\phi+1)n-1} + \frac{1}{2n} \cdot ((2\phi+1)n-1)-\frac{n-1}{n} }
      \\
     = \ & \frac{(\phi+1)n+(\phi-2)}{2\phi^2n} 
     \\
     = \ & \frac{1}{2} - \frac{(2-\phi)}{2\phi^2n} < \frac{1}{2},
    \end{split}    
\end{equation*}
where the first inequality also follows from $n-2-\frac{n^2-1}{2n}>0$ and $\epsilon \leq \frac{1}{(2\phi+1)n-1} < \frac{1}{n-1}$.

Therefore, in this case, this allocation is not MMA1 and MMAX to agent $n$, although her MMA1 and MMAX approximation ratio is maximized.
\end{proof}

\subsection{$0.618$-MMAX algorithm for symmetric agents}
In this part, we study the computation of approximate MMAX allocations when agents have identical weights. Actually, 0.618-MMAX allocations exist, which means the approximation ratio of MMA1 can be further improved to $0.618$. We can use the exact algorithm from the work of \cite{amanatidis2020multiple} and similar analysis to show that 0.618-MMAX allocations always exist. For completeness, we provide the analysis in detail. By Lemma \ref{lem: wmmax_reduction_approximation_goods}, it suffices to consider the ordered instance. Before moving to the analysis of Algorithm \ref{alg: round_robin_envy_cycle} (Algorithm 3 in \cite{amanatidis2020multiple}), we give the following two lemmas.

\begin{lemma}[\cite{amanatidis2018comparing}]
\label{lem: mms_nondecrease}
Suppose that $\mathcal{A} = (A_1,\ldots,A_n) \in \prod_{n}(M)$ is an $\MMS$ partition for agent $i$, for any $e \in A_j$ where $j \in N$, it holds that $\MMS_i(M\setminus \{e\},n-1) \geq \MMS_i(M, n)$.
\end{lemma}

\begin{lemma}
\label{lem: property_round_robin_envy_cycle}
In Algorithm \ref{alg: round_robin_envy_cycle}, before line \ref{lin: phi_mmax_first_round_2}, the partial allocation $\mathcal{A}^{p} =(A_1^{p},\ldots,A_n^{p})$ has the following properties:

\begin{enumerate}
    \item 
    \label{property_round_robin_envy_cycle_1}
    $A_i^{p}  = \{e_i\}$  for any $i \in N$.
    
    \item 
    \label{property_round_robin_envy_cycle_2}
    $v_i(A_i^{p}) > \phi \cdot v_i(e)$ for any $i \in L$ and $e \in M \ \setminus \cup_{k=1}^{n}A_{k}^{p}$. 
    
    \item \label{property_round_robin_envy_cycle_3}$\phi \cdot v_i(A_i^{p}) \geq v_i(A_j^{p})$ for any $i,j \in N \setminus L$.
    
     \item 
     \label{property_round_robin_envy_cycle_4}
     $v_i(A_i^{p}) \geq v_i(A_j^{p})$ and $v_j(A_i^{p}) \geq v_j(A_j^{p})$ for any $i, j \in N$ and $l_i < l_j$. 

\end{enumerate}
\end{lemma}

\begin{proof}
(1) In the while loop, each time agent $i$ is selected to choose an item, only one item $e_i$ is allocated to her in line \ref{lin: phi_mmax_first_round_1} or line \ref{lin: phi_mmax_first_round_2} of Algorithm \ref{alg: round_robin_envy_cycle}. If she obtains an item in line \ref{lin: phi_mmax_first_round_1} and then enters the set $L$, she will not lose it and cannot enter the loop again. If she obtains the item in line \ref{lin: phi_mmax_first_round_2}, she may lose her item due to line \ref{lin: phi_mmax_first_round_1} and gains the opportunity to enter the loop again to obtain a new $e_i$. When $N^{'} = \emptyset$, no agent will enter the loop, and each agent receives exactly one item $e_i$.

(2) In Algorithm \ref{alg: round_robin_envy_cycle}, as long as agent $i$ enters $L$, she will keep her item until the termination of the algorithm. Otherwise, she may come back to $N'$ and enters the loop again. Therefore, consider the time before she enters $L$. we can use $e_{old}$ to denote the item allocated to her currently. After she enters $L$, let $e_{new}$ denote the item she grabs from another agent in line \ref{lin: phi_mmax_first_round_1}. Clearly, we have $v_i(A_i^{p}) = v_i(e_i) = v_i(e_{new}) > \phi \cdot v_i(e_{old}) \geq \phi \cdot v_i(e)$ for any $e$ which is not allocated to any agent at the time. Meanwhile, we find that the number of items that are not allocated to any agent decreases with the execution of the algorithm. In the end, we can conclude that $v_i(A_i^{p}) > \phi \cdot v_i(e)$ for any $i \in L$ and $e \in M \ \setminus \cup_{k=1}^{n}A_{k}^{p}$.

(3) For any two agents $i,j \in N \setminus L$, these two agents may repeat the following procedure several times: enter the loop to obtain one item, lose their items and reenter again. Therefore, considering the last time each agent enters the loop, there are two cases:

\begin{itemize}
    \item[(a)] If agent $i$ enters the loop before agent $j$, it is easy to show $v_i(A_i^{p})=v_i(e_i) \geq v_i(e_j)=v_i(A_j^{p})$ since agent $i$ picks her favorite one from the remaining available items before agent $j$.
    \item[(b)] If agent $i$ enters the loop after agent $j$, because she does not enter $L$, we will have $\phi \cdot v_i(A_i^{p}) = \phi \cdot v_i(e_i) \geq v_i(e_j) = v_i(A_j^{p})$.
\end{itemize}

(4) If $i, j \in L$ and $l_i < l_j$, we know agent $j$ enters $L$ after agent $i$. In the algorithm, if one agent enters $L$, she will grab one item from another agent in $R$, and this item has the maximum value among items that are not allocated to agents in $L$ at that time. Hence, when agent $j$ enters $L$, she will have less choice than agent $i$, and we have $v_i(A_i^{p}) \geq v_i(A_j^{p})$. Note that we consider the ordered instance. In agent $j$'s view, the value of the item in agent $i$'s bundle will always be greater than that of her item. Thus, we have $v_j(A_i^{p}) \geq v_j(A_j^{p})$. If $i \in L$ and $j \in N back\setminus L$, similarly, agent $i$ has the priority to get one item that has the maximum value among the items which are not allocated to agents in $L$ at that time. For agent $j$, no matter whether she has been in $N\setminus L$ or enters $N \setminus L$ in the future, she only obtains one item from this set of items. Thus, we have $v_i(A_i^{p}) \geq v_i(A_j^{p})$ and $v_j(A_i^{p}) \geq v_j(A_j^{p})$. If $i, j \in N \setminus L$ and $l_i < l_j$, following the analysis of property \ref{property_round_robin_envy_cycle_3} in Lemma \ref{lem: property_round_robin_envy_cycle}, we have $v_i(A_i^{p}) \geq v_i(A_j^{p})$ and $v_j(A_i^{p}) \geq v_i(A_j^{p})$.
\end{proof}

\begin{algorithm}[h]
\caption{Envy-cycle elimination with preprocessing}
\label{alg: round_robin_envy_cycle}
\KwIn{An ordered instance $I'$($\lvert M \rvert = m$  items and $\lvert N \rvert = n$ agents}
\KwOut{An approximate MMAX allocation $\mathcal{A} = (A_1, \ldots, A_n)$}
Let $\mathcal{A} = (\emptyset, \ldots, \emptyset)$ and $\phi = \frac{1+\sqrt{5}}{2} \approx 1.618$

$L = \emptyset$; $N^{'} = N$;  $l_i = 0$ ($i = 1,\ldots,n$); $r=s=1$

\While{$N^{'} \neq \emptyset$}
{Let $i$ be the lexicographically first agent of $N^{'}$.

$e_i$ = $\arg \max_{e \in M}v_i(e)$ and $t_i$ = $m$ - $\lvert M \rvert$ +1

Let $R$ = ($N \setminus (N^{'} \cup L) $) $\cup \{i\}$ and $j$ = $\arg \max_{k \in R} v_i(e_k)$

\If {$\phi \cdot v_i(e_i) < v_i(e_j)$}
{
$A_j = A_j \setminus \{e_j\} $, $e_i = e_j$ and $A_i = A_i \cup \{e_i\} $ \label{lin: phi_mmax_first_round_1}

$L = L \cup \{i\}$

$l_i = r$ and $r = r+1$ 

$N^{'} = (N^{'} \setminus \{i\}) \cup \{j\}$
}
\Else
{
$N^{'} = N^{'} \setminus \{i\}$ and $M = M \setminus \{e_i\}$

$A_i = A_i \cup \{e_i\}$ \label{lin: phi_mmax_first_round_2}
}
}

\For{every $i \in N \setminus L$ in the decreasing order of $t_i$ }
{

$l_i = n-s+1$ and $s=s+1$

$e^{*} = \arg \max_{e \in M} v_i(e)$ 

$M = M \setminus \{e^{*}\}$ and $A_i = A_i \cup \{e^{*}\}$
}

Use the envy-cycle elimination algorithm to allocate the remaining items. \label{lin: envy_cycle_alg}

\Return $\mathcal{A}$
\end{algorithm}

\begin{theorem}
When agents have identical weights, Algorithm \ref{alg: round_robin_envy_cycle} can compute a $(\phi -1)$-MMAX allocation where $\phi \approx 1.618$ is the golden ratio.    
\end{theorem}

\begin{proof}

Let $\mathcal{A} = (A_1, \ldots,A_n)$ be the allocation returned by Algorithm \ref{alg: round_robin_envy_cycle}. If each agent obtains at most one item, it is easy to show that the allocation $\mathcal{A}$ is trivially MMAX. Thus, we only consider cases where there exist some agents who get more than one item, and there are two cases to be considered.

\underline{Case 1:} If all items are allocated to all agents before line \ref{lin: envy_cycle_alg} in Algorithm \ref{alg: round_robin_envy_cycle}, it is easy to conclude that for any agent $i \in L$, she has one item; for any agent $i \in N\setminus L$, she has at most two items.  Because we consider the ordered instance, for any agent $i \in N$, the item with the maximum value should be put in the bundle of the agent who enters $L$ first or one agent who has the smallest value of $l$ if $L = \emptyset$. Assume that there are $k$ ($1 \leq k \leq n-\lvert L \rvert$) bundles where each bundle has two items. Now we arbitrarily choose one agent $i \in N$ and let $e_{min}$ denote the item with the minimum value from $A_{-i}$. Next, consider the following two possible cases:

\begin{itemize}
    \item [(a)] If the bundle of agent $i$ has only one item, there are  $n + k -2$ items in $A_{-i} \setminus \{e_{min}\}$. By the pigeonhole principle, at least $n-k$ bundles contain one item in any $(n-1)$-partition of $A_{-i}\setminus \{e_{min}\}$ where no bundle is empty. According to property \ref{property_round_robin_envy_cycle_3} in Lemma \ref{lem: property_round_robin_envy_cycle}, in agent $i$'s perspective, there are at least $2k-1$ items where the value of each item is less than that of agent $i$'s bundle and at most $n-k-1$ items where the value of each item is greater than that of her bundle. Thus, there always exists one bundle with one item, and the value of it is less than that of agent $i$'s bundle.
    
    \item [(b)]  If the bundle of agent $i$ has two items, there are $n+k-3$ items in $A_{-i} \setminus \{e_{min}\}$. Similarly, based on the pigeonhole principle, there are at least $n-k+1$ bundles that only have one item in any $(n-1)$-partition of $A_{-i} \setminus \{e_{min}\}$ where no bundle is empty. Meanwhile, there are at least $k-2$ items where the value of each item is less than that of any item in her bundle and at most $n-1$ items where the value of each item is greater than that of any item in her bundle. Thus, there always exists one bundle with one item, and the value of it is less than that of agent $i$'s bundle.
    
\end{itemize}

Therefore, we can conclude that if the algorithm terminates before line \ref{lin: envy_cycle_alg}, the returned allocation $\mathcal{A}$ is MMAX.

\underline{Case 2:}  If all items are allocated to all agents after line \ref{lin: envy_cycle_alg} in Algorithm \ref{alg: round_robin_envy_cycle}, fix one agent $i \in N$ and choose another arbitrary agent $j \in N \setminus \{i\}$. If $\lvert A_j \rvert = 1$, we remove this item and agent $j$. By Lemma \ref{lem: mms_nondecrease}, we have $\MMS_{i}(A_{-i}\setminus\{A_j,\{e_{min}\}\}, n-2) \geq \MMS_{i}(A_{-i}\setminus \{e_{min}\},n-1)$ where $e_{min} = \arg \min_{e \in A_{-i}}v_i(e)$. If we can show the approximation ratio of MMAX in the reduced case, it also holds in the original case. Thus, w.l.o.g, we can assume that $\lvert A_j \rvert \geq 2$. Next, consider the occasion when the last item $e^{l}$ is added to $A_j$. At this time, $A_j$ may belong to agent $j^{'}$ instead of $j$ due to the possible swap in the envy-cycle elimination algorithm. Therefore, $A_i^{old}$ and $A_{j^{'}}^{old}$ are used to denote the bundles of agent $i$ and $j^{'}$, respectively, before item $e^{l}$ is allocated. After the execution of the envy-cycle elimination algorithm, we have $v_{i}(A_i) \geq v_i(A_i^{old})$. Now, there are two cases about where agent $i$ comes from when $e^{l}$ is allocated.

\begin{itemize}
    \item [(a)] If agent $i$ comes from $L$, there are two cases about $e^{l}$. If $e^{l}$ is added in line \ref{lin: phi_mmax_first_round_2}, we have $A_i^{old} = \{e_i\}$ and $A_j = A_{j^{'}}^{old} \cup \{e^{l}\} = \{e_{j^{'}},e^{l}\}$. In this case, by Lemma \ref{lem: property_round_robin_envy_cycle}, we have $v_i(A_i^{old}) \geq v_i(e_{j^{'}})$ and $v_i(A_i^{old}) \geq \phi \cdot v_i(e^{l})$. Now, if we combine them, we have $v_{i}(A_j) = v_{i}(e_{j^{'}}) + v_{i}(e^{l}) \leq (\phi^{-1}+1)v_i(A_i^{old}) \leq \phi \cdot v_i(A_i)$. Similarly, if $e^{l}$ is added in line \ref{lin: envy_cycle_alg}, by way of choosing who gets the new unallocated item in the envy-cycle elimination algorithm, we have $v_i(A_i^{old})\geq v_i(A_{j^{'}}^{old})$. Besides that, we also have $v_i(A_i^{old})\geq v_i(e_i) \geq \phi \cdot v_i(e^{l})$ by Lemma \ref{lem: property_round_robin_envy_cycle}. If we put them together, we get $v_i(A_j) \leq \phi \cdot v_i(A_i^{old}) \leq \phi \cdot v_i(A_i)$. Therefore, if agent $i$ comes from $L$, we have $v_i(A_j) \leq \phi \cdot v_i(A_i)$. Summing up respective inequalities for all $j \in N \setminus \{i\}$, we get $v_i(A_i) \geq \frac{1}{\phi} \cdot \frac{v_i(A_{-i})}{n-1}$. Because $e_{min}$ is allowed to be removed, we have $v_i(A_i) \geq \frac{1}{\phi} \cdot \frac{v_i(A_{-i}\setminus \{e_{min}\})}{n-1} \geq \frac{1}{\phi}\cdot \MMS_i(A_{-i} \setminus \{e_{min}\},n-1)$.
    
    \item [(b)] If agent $i$ comes from $N\setminus L$, there are three cases about $e^{l}$. If $e^{l}$ is added in line \ref{lin: phi_mmax_first_round_2}, we have $A_j = A_{j^{'}}^{old} \cup \{e^{l}\} = \{e_{j^{'}},e^{l}\}$. Next, we discuss the order of agent $i$ and $j^{'}$. If $l_i < l_{j^{'}}$, agent $i$ may not obtain the second item $e_{i}^{'}$ in line \ref{lin: phi_mmax_first_round_2} and we have $v_i(A_i) \geq v_i(A_i^{old}) \geq v_i(e_i) \geq v_i(e_{j^{'}}) \geq v_i(e^{l})$. If $l_i > l_{j^{'}}$, agent $i$ has obtained her second item $e_{i}^{'}$ and we will have $v_i(A_i) \geq v_i(A_i^{old}) \geq v_i(e_i) + v_i(e_i^{'}) \geq \frac{1}{\phi} \cdot v_i(e_{j^{'}}) + v_i(e^{l}) \geq \frac{1}{\phi} \cdot v_i(A_j)$ . If $e^{l}$ is added in line \ref{lin: envy_cycle_alg}, similarly we will have $v_i(A_i^{old}) \geq v_i(A_{j^{'}}^{old})$. In addition to that, we have $v_i(A_i^{old}) \geq v_i(e_i) + v_i(e_i^{'}) \geq 2 \cdot v_i(e^{l})$ by Lemma \ref{lem: property_round_robin_envy_cycle} and working principle between line \ref{lin: phi_mmax_first_round_2} and line \ref{lin: envy_cycle_alg}. Putting them together, we have $v_i(A_i) \geq v_i(A_i^{old}) \geq \frac{2}{3} \cdot v_i(A_{j^{'}}) \geq \frac{1}{\phi} \cdot v_i(A_{j^{'}})$. 

    From the above analysis, it can be seen that the bottleneck is when there exists at least one bundle $A_j$ where the last item is added in line \ref{lin: envy_cycle_alg} with $l_i < l_{j^{'}}$. If the value of the second added item in all these kinds of bundles is less than $(\phi -1)\cdot v_i(A_i)$, the whole analysis is simplified since this kind of bundle satisfies $v_i(A_j) = v_i(e_j^{'}) + v_i(e^{l}) \leq v_i(A_i) + (\phi -1)v_i(A_i) \leq \phi \cdot v_i(A_i)$. If there are some bundles where the value of the second item is strictly greater than $ (\phi - 1) \cdot v_i(A_i)$, the above trick does not work. To simplify the remaining analysis, we define one special bundle in $A_{-i}$. If $\lvert A_{j} \rvert = 2$, and the last item is added in line \ref{lin: envy_cycle_alg} with $l_i < l_{j^{'}}$ and the value of it is strictly greater than $(\phi - 1)\cdot v_i(A_i)$, we say that this bundle is $bad$ and all items in the $bad$ bundle are $bad$ items. 
    
    Now we need to know some extra information about these $bad$ bundles. Assume that there are $k$ $bad$ bundles. Notice that the order of these bundles $l_j$ are the last $k$ consecutive numbers since the second added item is allocated in the decreasing order of $l_j$ in $N \setminus L$ and the value of the second added item should be less than $\frac{1}{2}v_i(A_i)$ if this agent picks another item in the execution of the envy-cycle elimination algorithm. At the same time, there are $n-1-k$ items in $A_{-i}$ where the value of each item is greater than the value of each first added $bad$ item in $bad$ bundles, and we say that these $n-1-k$ items are $big$ items.
    
    Let $\mathcal{B} = \{B_j\}_{j \neq i}$, where $\min_{j \in N \setminus \{i\}} v_i(B_j) = \MMS_{i}(A_{-i} \setminus \{e_{min}\}, n-1)$. denote an $(n-1)$-partition of $A_{-i} \setminus \{e_{min}\}$.  Suppose, for the contradiction, that $ \MMS_{i}(A_{-i} \setminus \{e_{min}\}, n-1) > \phi \cdot v_i(A_i)$. Now we have $n-1-k$ $big$ items and $2k$ $bad$ items in $A_{-i} \setminus \{e_{min}\}$. Next, consider the distribution of $big$ and $bad$ items in $\mathcal{B}$. By the pigeonhole principle, we can always find $k$ bundles in $\mathcal{B}$ where the total number of $big$ items and $bad$ items is at least $2k$. In the remaining $n-1-k$ bundles, because the value of one bundle which is not a $bad$ bundle in $A_{-i}$ is less than $\phi \cdot v_i(A_i)$, and the number of items in the remaining $n-1-k$ bundles in $\mathcal{B}$ is less than or equal to the number of items in the $n-1-k$ bundles that do not contain any $bad$ items in $A_{-i}$, the mean value of these bundles is less than $\phi\cdot v_i(A_i)$, which contradicts the assumption that the value of each bundle in $\mathcal{B}$ is strictly greater than $\phi \cdot v_i(A_i)$.   
\end{itemize}
Therefore, if the algorithm terminates after line \ref{lin: envy_cycle_alg}, the returned allocation $\mathcal{A}$ is $(\phi-1)$-MMAX.
\end{proof}

\end{document}